\PassOptionsToPackage{table,xcdraw}{xcolor}
\documentclass[10pt,journal,compsoc]{IEEEtran}

\usepackage{amssymb}  

\newenvironment{proof}[1][Proof]{\par
  \noindent\textit{#1.} \ignorespaces
}{%
  \hfill$\blacksquare$\par
}

\usepackage{scalerel}
\usepackage{tikz}
\usetikzlibrary{svg.path}

\definecolor{orcidlogocol}{HTML}{A6CE39}
\tikzset{
  orcidlogo/.pic={
    \fill[orcidlogocol] svg{M256,128c0,70.7-57.3,128-128,128C57.3,256,0,198.7,0,128C0,57.3,57.3,0,128,0C198.7,0,256,57.3,256,128z};
    \fill[white] svg{M86.3,186.2H70.9V79.1h15.4v48.4V186.2z}
                 svg{M108.9,79.1h41.6c39.6,0,57,28.3,57,53.6c0,27.5-21.5,53.6-56.8,53.6h-41.8V79.1z M124.3,172.4h24.5c34.9,0,42.9-26.5,42.9-39.7c0-21.5-13.7-39.7-43.7-39.7h-23.7V172.4z}
                 svg{M88.7,56.8c0,5.5-4.5,10.1-10.1,10.1c-5.6,0-10.1-4.6-10.1-10.1c0-5.6,4.5-10.1,10.1-10.1C84.2,46.7,88.7,51.3,88.7,56.8z};
  }
}

\newcommand\orcidicon[1]{\href{https://orcid.org/#1}{\mbox{\scalerel*{
\begin{tikzpicture}[yscale=-1,transform shape]
\pic{orcidlogo};
\end{tikzpicture}
}{|}}}}

\usepackage{amsmath,amsfonts}
\usepackage{algorithmic}
\usepackage{algorithm}
\usepackage{array}
\usepackage[caption=false,font=normalsize,labelfont=sf,textfont=sf]{subfig}
\usepackage{textcomp}
\usepackage{stfloats}
\usepackage{url}
\usepackage{verbatim}
\usepackage{graphicx}
\usepackage{cite}
\usepackage[hidelinks]{hyperref}

\usepackage{graphicx} 
\usepackage[shortlabels]{enumitem}
\usepackage{amsmath}
\usepackage{braket}
\usepackage{subcaption}
\usepackage{stfloats}
\usepackage{float}
\usepackage{etoolbox}
\usepackage{multirow}
 \usepackage{multirow}
\usepackage[normalem]{ulem}
\useunder{\uline}{\ul}{}

\usepackage{colortbl} 
\usepackage{hhline}
\usepackage{pgfplots} 
\usepackage{array}    
\usepackage{booktabs} 
\usepackage{graphicx} 
\usepackage{color}
\definecolor{bittersweet}{rgb}{1.0, 0, 0.5}

\newcommand{\ZH}[1]{{\color{black}#1}}

\pgfplotsset{compat=1.17} 
\BeforeBeginEnvironment{wrapfigure}{\setlength{\intextsep}{-5pt}\setlength{\columnsep}{10pt}}
\usepackage{graphicx}
\usepackage{wrapfig}
\usepackage{adjustbox}
\usepackage{mathtools}

\newtheorem{theorem}{Theorem}

\usepackage{makecell}
\usepackage{booktabs}
\usepackage{color}

\begin{document}

\title{OffsetCrust: Variable-Radius Offset Approximation with Power Diagrams}

\author{
Zihan Zhao \orcidicon{0009-0003-5962-4870}, Pengfei Wang \orcidicon{0000-0002-2079-275X}, Minfeng Xu \orcidicon{0000-0002-6553-5191}, Shuangmin Chen \orcidicon{0000-0002-0835-3316}, Shiqing Xin* \orcidicon{0000-0001-8452-8723}, Changhe Tu \orcidicon{0000-0002-1231-3392}, and Wenping Wang \orcidicon{0000-0002-2284-3952} ~\IEEEmembership{Fellow,~IEEE}

\IEEEcompsocitemizethanks{
\IEEEcompsocthanksitem Zihan Zhao, Pengfei Wang, Shiqing Xin, and Changhe Tu are with the School of Computer Science and Technology, Shandong University, Qingdao 266237, China (e-mail: \href{mailto:zihanzhao2000@gmail.com}{zihanzhao2000@gmail.com}; \href{pengfei1998@foxmail.com}{pengfei1998@foxmail.com}; \href{mailto:xinshiqing@sdu.edu.cn}{xinshiqing@sdu.edu.cn}; \href{mailto:chtu@sdu.edu.cn}{chtu@sdu.edu.cn}).
\IEEEcompsocthanksitem Minfeng Xu is with the School of Computer Science and Technology, Shandong University of Finance and Economics, Jinan 250014, China (e-mail: \href{mailto:mfxu_sdu@163.com}{mfxu\_sdu@163.com}).
\IEEEcompsocthanksitem Shuangmin Chen is with the School of Information and Technology, Qingdao University of Science and Technology, Shandong 266101, China,
and with the Shandong Key Laboratory of Deep Sea Equipment Intelligent Networking, Qingdao, China
(e-mail: \href{mailto:csmqq@163.com}{csmqq@163.com}).
\IEEEcompsocthanksitem Wenping Wang is with the Department of Computer Science and Engineering, Texas A\&M University, College Station, TX 77843 USA (e-mail:
\href{mailto:wenping@tamu.edu}{wenping@tamu.edu}).
\IEEEcompsocthanksitem *Shiqing Xin is the corresponding author.
\IEEEcompsocthanksitem Manuscript received April 19, 2021; revised August 16, 2021.
}
\markboth{Journal of \LaTeX\ Class Files,~Vol.~14, No.~8, August~2021}%
{Shell \MakeLowercase{\textit{et al.}}: A Sample Article Using IEEEtran.cls for IEEE Journals}
}

\maketitle

\begin{abstract}
Offset surfaces, defined as the Minkowski sum of a base surface and a rolling ball, play a crucial role in geometry processing, with applications ranging from coverage motion planning to brush modeling. While considerable progress has been made in computing \emph{constant-radius} offset surfaces, computing \emph{variable-radius} offset surfaces remains a challenging problem.
In this paper, we present \emph{OffsetCrust}, a novel framework that efficiently addresses the variable-radius offsetting problem by computing a power diagram. Let~$\mathcal{R}$ denote the radius function defined on the base surface~$\mathcal{S}$. The power diagram is constructed from contributing sites, consisting of carefully sampled base points on~$\mathcal{S}$ and their corresponding off-surface points, displaced along $\mathcal{R}$-dependent directions. In the constant-radius case only, these displacement directions align exactly with the surface normals of~$\mathcal{S}$.
Moreover, our method mitigates the misalignment issues commonly seen in crust-based approaches through a lightweight fine-tuning procedure. We validate the accuracy and efficiency of {OffsetCrust} through extensive experiments, and demonstrate its practical utility in applications such as reconstructing original boundary surfaces from medial axis transform~(MAT) representations.
\end{abstract}


\section{Introduction}
\label{sec:introduction}

Offset surfaces, defined as the Minkowski sum of a base surface and a rolling ball~\cite{rossignac1986offsetting}, with either constant or varying radius, play a fundamental role in geometry processing. They support a wide range of applications, including motion planning~\cite{singh2011robot} and brush-stroke simulation~\cite{held2021weighted}, as illustrated in Figure~\ref{fig:intro_offset}(a,c).

Given a base surface~$\mathcal{S}$ and a radius function~$\mathcal{R}$ defined on~$\mathcal{S}$, the offsetting problem can be formulated as an iso-surface extraction problem. We define a generalized distance function as follows:
\begin{equation}
    \phi(x) \coloneqq \min_{p \in \mathcal{S}} \left( \|x - p\| - \mathcal{R}(p) \right),
\label{eq:offset_distance_field}
\end{equation}
and extract the offset surface~$\mathcal{S}_\mathcal{R}^\text{off}$ as the zero level set, i.e.,~$\phi(x) = 0$.  
In particular, when~$\mathcal{S}$ is closed and orientable,~$\mathcal{S}_\mathcal{R}^\text{off}$ can be further divided into inward and outward offset layers.

\begin{figure}[!t]
    \centering
    \includegraphics[width=.98\linewidth]{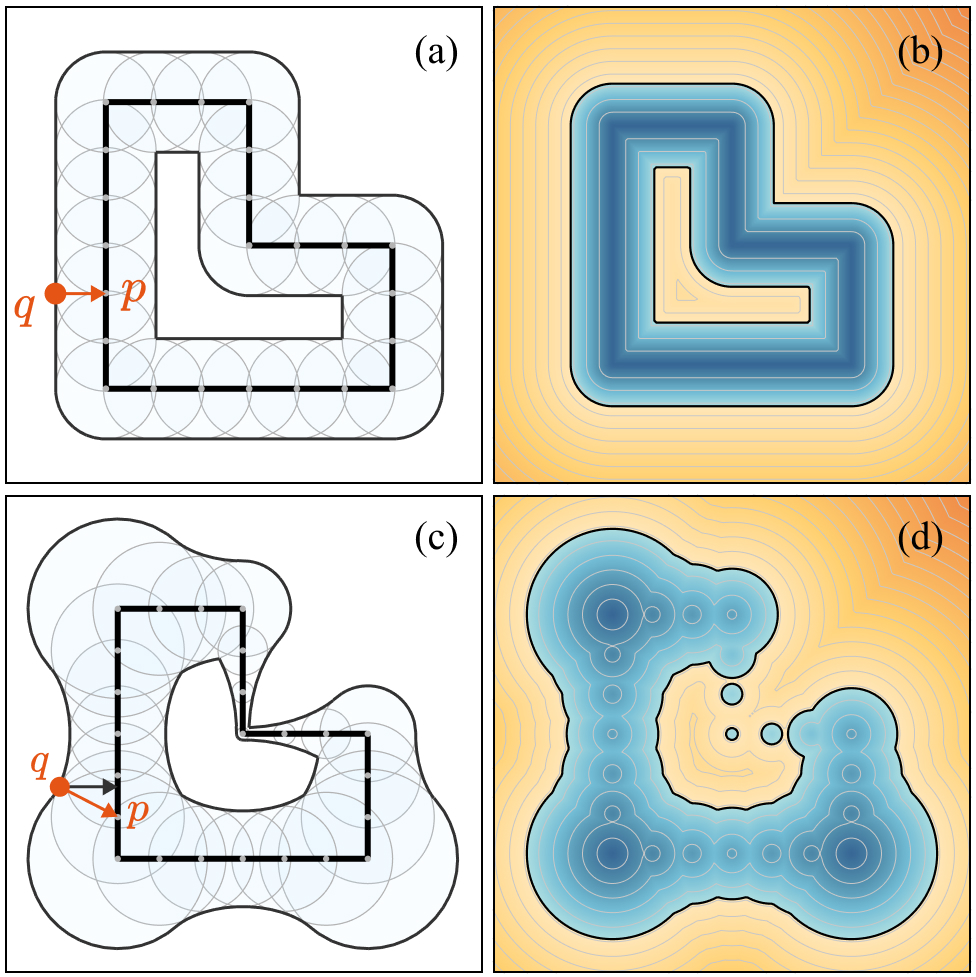}
    \caption{
    Offsetting a solid ``L''. 
    (a) Constant-radius offset: for a point \( q \), the minimizer \( p \) coincides with the closest point on the surface. 
    (b) The offset surface can be defined as the iso-surface of the corresponding distance function. 
    (c) Variable-radius offset: for a point \( q \), the minimizer \( p \) may differ from the closest point. 
    (d) Approximating Eq.~(\ref{eq:offset_distance_field}) by simply replacing the continuous boundary with a set of sample points leads to oscillation artifacts. Moreover, this na\"{i}ve approximation incurs a high computational cost.
    }
    \label{fig:intro_offset}
\end{figure}

As shown in Figure~\ref{fig:intro_offset}(a,b), a straightforward yet important observation is that, in the constant-radius case, the nearest point on~$\mathcal{S}$ serves as the minimizer of the generalized distance. As a result, the offset surface coincides exactly with the iso-surface of the unsigned distance function of~$\mathcal{S}$, evaluated at a fixed offset value. This property underpins most existing approaches~\cite{lorensen1998marching, chernyaev1995marching, ju2002dual, nielson2004dual, chen2019half, chen2021neural, chen2022neural, doi1991efficient, treece1999regularised, shen2021deep} to constant-radius offsetting. Given the significance of this task, several recent methods~\cite{wang2024pco, zint2023feature} have been proposed to improve surface quality and runtime performance. However, these methods are not applicable to the more general and challenging case of variable-radius offsetting.

As shown in Figure~\ref{fig:intro_offset}(c), we assign a non-uniform radius function to the boundary and visualize the resulting true variable-radius offset. In this scenario, for a point \( q \), the minimizer \( p \) may differ from the closest point on the surface, typically deviating from the surface normal at an angle. A na\"{i}ve approach to approximate Eq.~(\ref{eq:offset_distance_field}) is to replace the continuous boundary with a set of sampled points. However, this approximation introduces oscillation artifacts (see Figure~\ref{fig:intro_offset}(d)) and incurs a high computational cost.

In previous research, very few works have explicitly tackled this challenge. Some existing methods define variable-radius offsets differently, by specifying the displacement amount along surface normals. These approaches are typically formulated on parametric curves and surfaces~\cite{elber1997comparing, maekawa1999overview, pham1992offset, kim1993approximation}. However, they are not applicable when only a radius function is provided and must also contend with the complex issue of self-intersections~\cite{wallner2001self}.

The goal of this paper is to develop an efficient algorithm for the variable-radius offsetting problem that produces geometrically and topologically accurate results. Moreover, the sharp features inherent in the offset surface should be faithfully preserved. To this end, we propose \emph{OffsetCrust}, a crust-based framework designed to achieve these objectives. Our method proceeds as follows: we sample a sufficient number of points from the base surface and displace them along specially designed directions, computed based on the variation of the radius function~$\mathcal{R}$. By assigning different weights to base points and displaced points, the variable-radius offset can be effectively approximated as the subset of power diagram facets that separate these two types of sites (see Section~\ref{sec:insight}).

We further introduce a carefully designed sampling strategy (see Section~\ref{sec:sample_methods}) and a refinement technique (see Section~\ref{sec:vertex_refinement}) to address the misalignment issues inherent in crust-based methods. We validate the effectiveness of our approach through extensive experiments (see Section~\ref{sec:evaluation}) and demonstrate its utility in several modeling tasks (see Section~\ref{sec:applications}). Our code is available at \href{https://github.com/zih-an/offsetcrust.git}{\nolinkurl{https://github.com/zih-an/offsetcrust.git}}.

\section{Related Work}

\subsection{Constant-Radius Offset}
There is a large body of literature addressing the problem of constant-radius offsets. We categorize these approaches into four main groups.

\textbf{Parametric Offsets.}
Parametric curves and surfaces are widely used in CAD/CAM as continuous shape representations. In this setting, both the base surface and its offset can be analytically expressed, typically using B-spline curves or NURBS surfaces. 
A constant-radius offset displaces each point along the surface normal by a fixed distance \( d \), given by \( \hat{\mathbf{r}}(t) = \mathbf{r}(t) + d\mathbf{n}(t) \)~\cite{maekawa1999overview, pham1992offset, elber1997comparing}.  
However, this method requires tedious handling of self-intersections.

\textbf{Voxelization Representations.}
Constant-radius offsets can also be computed via the Minkowski sum of the shape with a discrete sphere. 
Li et al.~\cite{li2010gpu, li2014sweep, li2011voxelized} proposed GPU-accelerated approximations using voxelization. 
Beyond the Minkowski sum, Chen et al.~\cite{chen2019half} employed dexel structures~\cite{van1986real} for parallelizable offset computation. 
While voxel-based methods are easy to implement, they often suffer from high computational cost, limited accuracy, and the loss of sharp features.

\textbf{Distance Fields.}
Constant-radius offsets can be represented as level sets of signed distance fields (SDFs). Given a well-computed SDF (typically defined over a grid), the offset surface becomes an iso-surface extraction problem. Numerous works have explored this approach~\cite{lorensen1998marching, chernyaev1995marching, ju2002dual, nielson2004dual, chen2021neural, chen2022neural, doi1991efficient, treece1999regularised, shen2021deep}.  
To improve accuracy and preserve sharp features, various enhancements have been proposed. 
Qu et al.~\cite{qu2004feature} used offset distance fields with irregular, dense grids. 
Pavic et al.~\cite{pavic2008high} employed a variant of Dual Contouring~\cite{ju2002dual}, although their method struggled in concave regions. 
Zint et al.~\cite{zint2023feature} combined topology-adapted octrees with Dual Contouring and remeshing to produce high-quality meshes, though self-intersections remained unresolved.

\textbf{Generalized Voronoi Diagrams.}
Polygonal curves and surfaces are commonly used as discrete shape representations. 
Generalized Voronoi diagrams, constructed using line segments, circles, and arcs as sites, have been applied to compute constant-radius offsets for 2D polygons~\cite{Menelaos2004, held1991computational, held2016generalized, held2021weighted}.  
These diagrams partition the plane into curved cells, with the offset in each cell determined solely by its associated site. 
However, extending such diagrams to 3D remains a significant challenge.

\textbf{Point Cloud Based.}
Point clouds can also be directly used to compute morphological operations~\cite{calderon2014point}. Although such methods avoid oscillations via moving least squares, they may still introduce approximation errors and do not account for the piecewise-linear nature of mesh representations.

\subsection{Variable-Radius Offset}
Compared to constant-radius offsets, relatively few algorithms have been developed for variable-radius offsets, most of which focus on parametric inputs.  
Kim et al.~\cite{kim1993approximation} extended the concept of offsets to parametric surfaces with variable radii, but their approach is not applicable when only a radius function is given.  
Qun et al.~\cite{qun1997variable} introduced explicit formulas for computing the envelope of balls, enabling exact representation of variable-radius offsets—but at the cost of handling complex self-intersections.

For general triangle meshes, a common strategy approximates the offset surface as the union of primitives: spheres at vertices, cones along edges, and slabs on faces.  
Due to the high computational cost of this strategy, Woerl et al.~\cite{woerl2020variable} proposed a discretized, GPU-accelerated volumetric method, which still suffers from geometric inaccuracy.

Apart from the traditional definition of offsets, several variants have been proposed recently. Mitered Offset~\cite{cao2024robust} preserves all convex sharp features rather than producing the rounded shape typical of classical formulations. 
Topological Offset~\cite{zint2025topological} automatically adapts the offset distance locally, effectively producing a spatially varying offset that preserves correct topology.

\subsection{Crust-Based Methods} 
Crust-based methods typically rely on pairs of seed points to define a surface that separates one set of seeds from another, generally requiring the explicit construction of a Voronoi or power diagram.  
Careful seed placement is essential to prevent misaligned facets.

PowerCrust~\cite{amenta2001power} uses the Voronoi diagram of sample points to approximate the medial axis transform (MAT), identifying interior and exterior poles. The power diagram of these poles is then used to reconstruct the surface and approximate constant-radius offsets.  
VoroCrust~\cite{abdelkader2020vorocrust, abdelkader2018sampling} improves robustness by employing co-circular or co-spherical sampling strategies to reduce misalignment.

\begin{figure*}[!tbp]
    \centering
    \includegraphics[width=.99\linewidth]{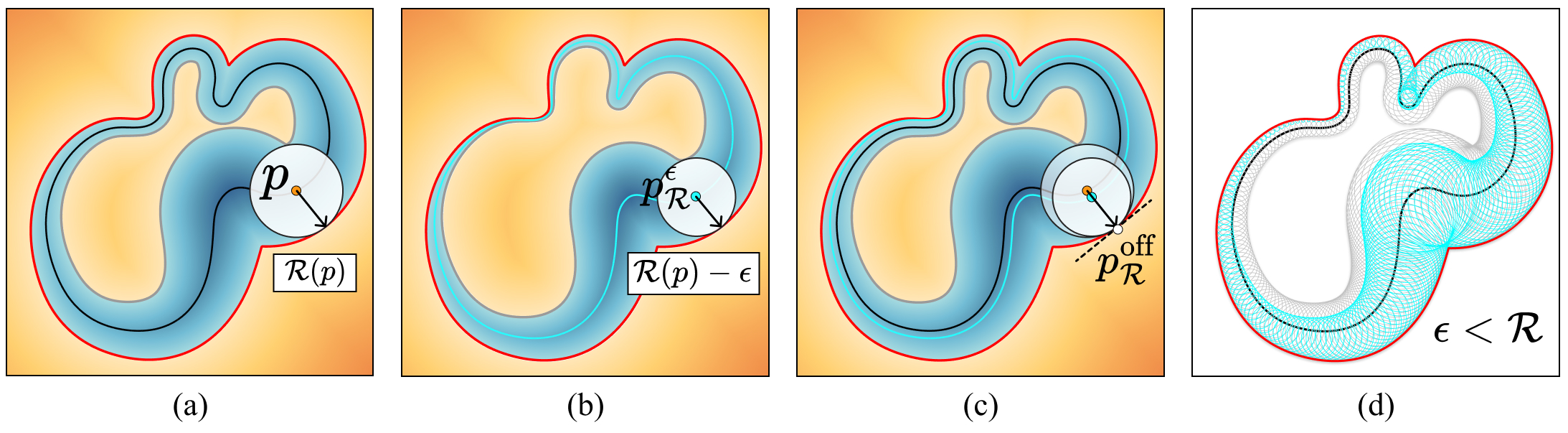}
    \caption{Insight into why the power diagram helps determine the variable-radius offset.  
(a) A point $p$ on the base surface;  
(b) A displaced point $p_{\mathcal{R}}^{\epsilon}$, located along the $\mathcal{R}$-dependent direction at a small distance $\epsilon$;  
(c) The three points $p$, $p_{\mathcal{R}}^{\epsilon}$, and the corresponding point on the offset surface $p_{\mathcal{R}}^{\text{off}}$ are colinear;  
(d) Visualization of the balls centered at the base sample points (gray)
and those centered at the displaced points (cyan). The value $\epsilon$ can be chosen sufficiently small, or at least not exceeding the specified radius.
In the power diagram, the bisector between $p$ and $p_{\mathcal{R}}^{\epsilon}$  
typically defines the tangent line at $p_{\mathcal{R}}^{\text{off}}$.
}
\vspace{-2mm}
\label{fig:insight_sdf}
\end{figure*}

\section{Insight}
\label{sec:insight}

We explain why power diagrams can be used to solve the variable-radius offsetting problem.  
Suppose we are given a smooth surface~$\mathcal{S}$, typically discretized as a triangular mesh~\(\mathcal{M}\),  
along with a user-specified radius function~$\mathcal{R}$.

\subsection{Continuous Setting}
Given a smooth, closed, and orientable surface~$\mathcal{S}$,  
and any positive radius function~$\mathcal{R} > 0$,  
the offset surface~$\mathcal{S}_{\mathcal{R}}^{\text{off}}$ 
consists of inward and outward layers.  
Each point~$p \in \mathcal{S}$ induces a ball centered at~$p$ with radius~$\mathcal{R}(p)$.  
As the offset surface is the envelope of these rolling balls,  
any point~$q$ on the surface of the ball centered at~$p$  
either lies on the final offset surface or is occluded by another ball.

As shown in Figure~\ref{fig:insight_sdf},  
assume a point~$p \in \mathcal{S}$ contributes to~$\mathcal{S}_{\mathcal{R}}^{\text{off}}$,  
yielding a corresponding offset point~$p_{\mathcal{R}}^{\text{off}}$. 
The segment~$\overline{pp_{\mathcal{R}}^{\text{off}}}$ lies entirely within  
the band between the inward and outward layers.  
The direction of this segment, denoted~$\boldsymbol{n}_{\mathcal{R}}^{\text{off}}(p)$,  
may differ from the surface normal at~$p$.

Let $\epsilon > 0$ be sufficiently small, and define:
\begin{equation}
p_{\mathcal{R}}^{\epsilon} := p + \epsilon \cdot \boldsymbol{n}_{\mathcal{R}}^{\text{off}}(p).
\end{equation}
Then, the three points $p$, $p_{\mathcal{R}}^{\epsilon}$, and $p_{\mathcal{R}}^{\text{off}}$ are colinear.  
We thus obtain:
\begin{equation}
\|p_{\mathcal{R}}^{\text{off}} - p\| - \mathcal{R}(p) = \|p_{\mathcal{R}}^{\text{off}} - p_{\mathcal{R}}^{\epsilon}\| - (\mathcal{R}(p) - \epsilon) = 0,
\label{eq:offset_union_distance_field_filter}
\end{equation}
or equivalently:
\begin{equation}
\|p_{\mathcal{R}}^{\text{off}} - p\|^2 - \mathcal{R}^2(p) = \|p_{\mathcal{R}}^{\text{off}} - p_{\mathcal{R}}^{\epsilon}\|^2 - (\mathcal{R}(p) - \epsilon)^2 = 0.
\label{eq:offset_union_distance_field_filter:square}
\end{equation}

This motivates the construction of two groups of points and the use of a power diagram to determine the offset surface; see Figure~\ref{fig:insight_sdf}.
In particular, the bisector between $p$ and $p_{\mathcal{R}}^{\epsilon}$ typically defines the tangent line at $p_{\mathcal{R}}^{\text{off}}$. As for the choice of~$\epsilon$, its value can be made sufficiently small, or at most equal to the specified radius.

\subsection{Discrete Representation via Power Diagram}
\label{sec:discrete_impl}

\begin{figure}[!tbp]
    \centering
    \includegraphics[width=.7\linewidth]{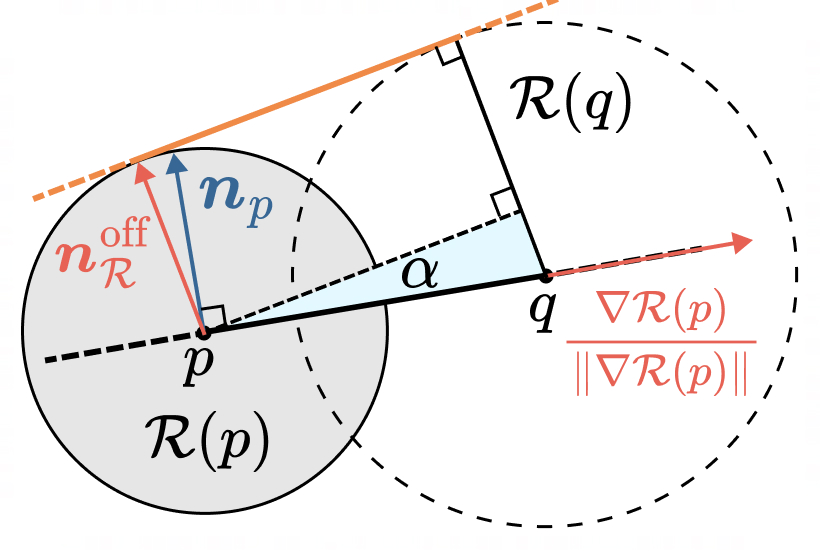}
    \caption{Proof to Theorem~\ref{thm:Displacement}.
    }
    \label{fig:variable_normal_deviate}
\end{figure}

\textbf{Power Diagram.}
Given a set of weighted points (also called \emph{sites})
\begin{equation}
\left\{(p_i, w_i) \mid p_i \in \mathbb{R}^3,\ w_i \in \mathbb{R} \right\}_{i=1}^n,
\end{equation}
the power diagram~\cite{aurenhammer1987power} partitions~\(\mathbb{R}^3\) into (possibly unbounded) convex regions.  
Each region~\(\Omega_i\), corresponding to a site~\((p_i, w_i)\), contains all points~\(x \in \mathbb{R}^3\) minimizing the {power distance}:
\begin{equation}
d^{\mathrm{pow}}_{p_i}(x) = \|x - p_i\|^2 - w_i.
\end{equation}
In our setting, we construct the power diagram using:
\begin{itemize}
    \item Points sampled from the base surface, each with weight~$\mathcal{R}^2(p)$;
    \item Displaced points offset in the direction~$\boldsymbol{n}_{\mathcal{R}}^{\text{off}}(p)$,  
    each with weight~$(\mathcal{R}(p) - \epsilon)^2$.
\end{itemize}
The computation of~$\boldsymbol{n}_{\mathcal{R}}^{\text{off}}(p)$ is described next.

\textbf{Displacement Direction.}
As previously noted, the displacement direction~$\boldsymbol{n}_{\mathcal{R}}^{\text{off}}(p)$  
may not align with the surface normal.  
Instead, it depends on the radius function~$\mathcal{R}$.

\begin{theorem}[Displacement Direction]\label{thm:Displacement}
If $\|\nabla \mathcal{R}(p)\| \leq 1$, the offset direction~$\boldsymbol{n}_{\mathcal{R}}^{\text{off}}(p)$ at point~$p$ can be constructed by rotating the unit normal~$\boldsymbol{n}_p$ by an angle~$\alpha = \arcsin(\|\nabla \mathcal{R}(p)\|)$ around the axis
\begin{equation}
\frac{\nabla \mathcal{R}(p)}{\|\nabla \mathcal{R}(p)\|} \times \boldsymbol{n}_p.
\end{equation}
\end{theorem}
\begin{proof}
As illustrated in Figure~\ref{fig:variable_normal_deviate},  
let $p$ and $q$ be a pair of nearby points on the surface.  
Each point defines a sphere centered at itself,  
with radii $\mathcal{R}(p)$ and $\mathcal{R}(q)$, respectively.

We consider the right triangle highlighted in Figure~\ref{fig:variable_normal_deviate}.  
The hypotenuse has length~$\|pq\|$, and one leg has length~$|\mathcal{R}(p) - \mathcal{R}(q)|$.  
This leg lies in the direction of the gradient~$\nabla \mathcal{R}(p)$,  
while the other leg lies in the tangent plane at the offset point~$p_{\mathcal{R}}^{\text{off}}$.  
Given this geometric configuration, the angle~$\alpha$, and hence the displacement direction, can be computed directly.

\end{proof}


\begin{figure}[!t]
\centering
\includegraphics[width=0.99\linewidth]{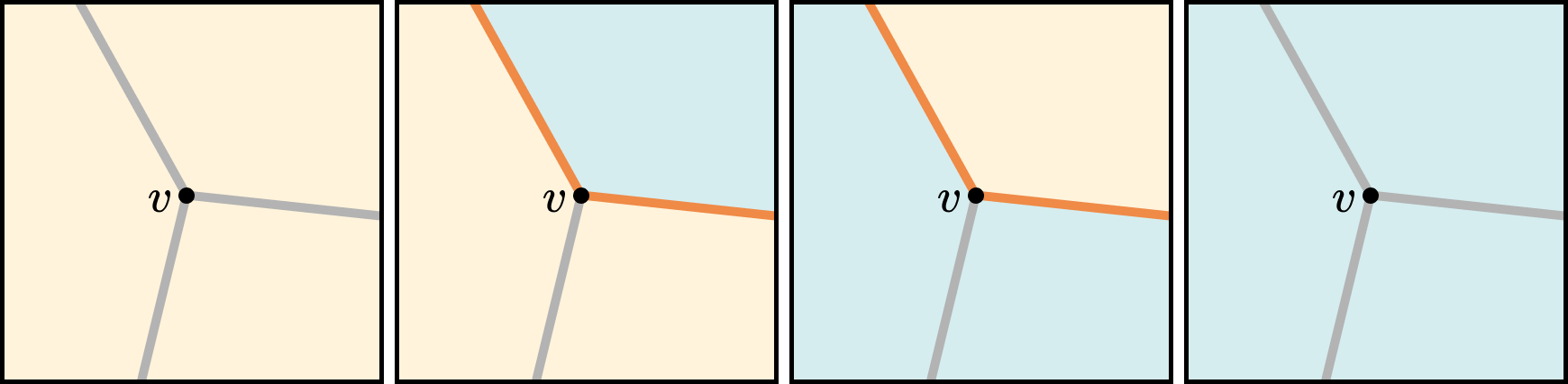}
\caption{
In a 2D power diagram, a vertex typically has a degree of three in non-degenerate cases. 
Each power cell is labeled as $C_p$ \textit{(blue)} or $C_{p^{\epsilon}}$ \textit{(yellow)}, and the offset surface \textit{(orange)} is extracted as the separator between $C_p$ and $C_{p^{\epsilon}}$. 
This figure displays all four labeling configurations. For any of the configurations, non-manifold artifacts cannot occur in the separator.
Note that the vertex does not contribute to the offset surface for the first and last configurations.}
\label{fig:vor_colored}
\end{figure}

\textbf{Manifold Output Without the Need for Explicit Self-Intersection Handling.}
Without loss of generality,  
we assume the power diagram does not exhibit degeneracies—specifically, no more than three sites determine a Voronoi vertex in 2D, and no more than four in 3D. This matches the behavior of the CGAL implementation, which guarantees a unique triangulation even under co-spherical degeneracies~\cite{cgal:triangulation}.

Figure~\ref{fig:vor_colored} shows four representative configurations,  
where the regions corresponding to base points and displaced points are visualized in different colors.  
It is straightforward to verify that non-manifoldness cannot occur by examining all possible configurations.

Furthermore,  
unlike traditional offsetting methods that require explicit handling of self-intersections,  
our approach avoids this complication entirely.  
Let $p$ be a sample point on the base surface that does not contribute to the final offset surface—that is, its corresponding offset point $p_{\mathcal{R}}^{\text{off}}$ lies strictly inside the envelope.  
We now explain why such a point cannot contribute to the crust,  
and thus requires no special treatment.

Since $p_{\mathcal{R}}^{\text{off}}$ lies within the interior of the envelope,  
there must exist another point $q$ whose associated ball contains $p_{\mathcal{R}}^{\text{off}}$.  
This implies:
\begin{equation}
    \|q - p_{\mathcal{R}}^{\text{off}}\| < \mathcal{R}(q),
\end{equation}
or, equivalently,
\begin{equation}
    \|q - p_{\mathcal{R}}^{\text{off}}\| - \mathcal{R}(q) < 0.
\end{equation}
Therefore, we further have
\begin{equation}
    \|p - p_{\mathcal{R}}^{\text{off}}\| - \mathcal{R}(p) = 0 > \|q - p_{\mathcal{R}}^{\text{off}}\| - \mathcal{R}(q),
\end{equation}
or 
\begin{equation}
    \|p - p_{\mathcal{R}}^{\text{off}}\|^2 - \mathcal{R}^2(p) = 0 > \|q - p_{\mathcal{R}}^{\text{off}}\|^2 - \mathcal{R}^2(q),
\end{equation}
or 
\begin{equation}
d_{p}^\text{pow}(p_{\mathcal{R}}^{\text{off}})>d_{q}^\text{pow}(p_{\mathcal{R}}^{\text{off}}),
\end{equation}
which shows that $p_{\mathcal{R}}^{\text{off}}$ must reside in the power cell of $q$  
(or of another, more dominant site).  
To summarize, the point $p$ must be located in the interior of the envelope and thus does not contribute to the crust.

\begin{figure}[!tbp]
\centering
\includegraphics[width=.7\linewidth]{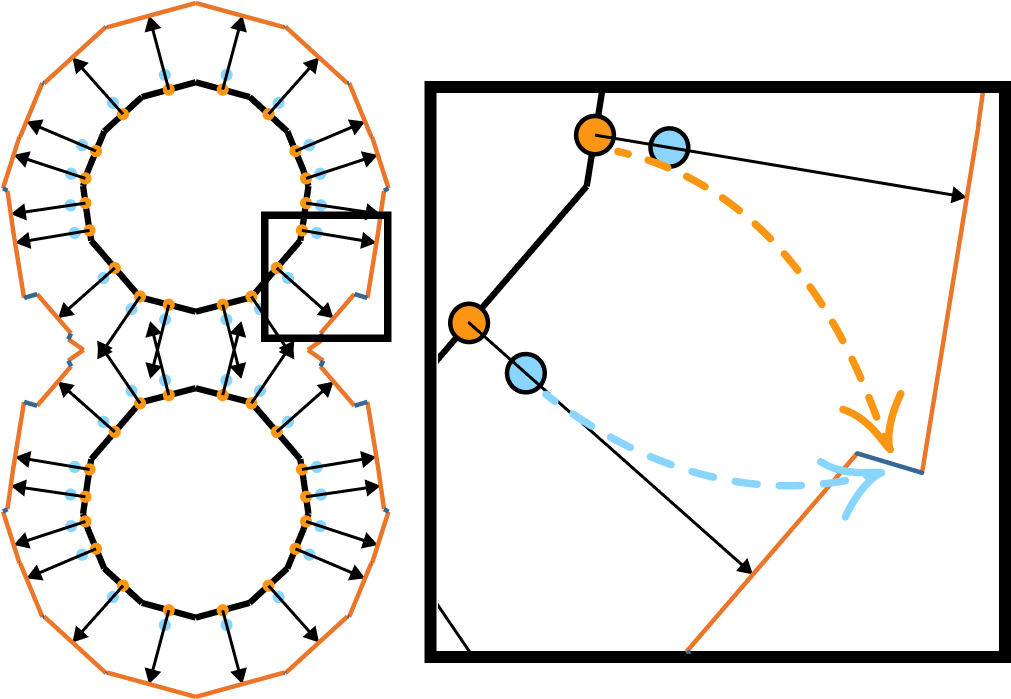}
\caption{
Crust-based approaches are prone to misalignment issues,  
where the misaligned segment (shown in blue) is formed by a base point  
and the displaced point of a different base point.
}
\label{fig:offset 2d_crust_and_smooth}
\end{figure}

\textbf{Misaligned Facets.}
As shown in Figure~\ref{fig:offset 2d_crust_and_smooth},  
crust-based approaches are prone to misalignment issues.  
Suppose $p$ and $q$ are a pair of nearby points.  
Point $p$ (resp. $q$) induces a facet with surface normal~$\boldsymbol{n}_{\mathcal{R}}^{\text{off}}(p)$ (resp. $\boldsymbol{n}_{\mathcal{R}}^{\text{off}}(q)$)  
at the offset point~$p_{\mathcal{R}}^{\text{off}}$ (resp. $q_{\mathcal{R}}^{\text{off}}$).  
However, a gap may exist between the two facets,  
often filled by the bisector defined by either the pair $(p, q_{\mathcal{R}}^{\epsilon})$  
or the pair $(q, p_{\mathcal{R}}^{\epsilon})$.

VoroCrust~\cite{abdelkader2020vorocrust, abdelkader2018sampling} addresses such misalignments  
by carefully positioning the sample points.  
However, it is limited to representing the base surface itself.  
In our scenario, we instead use a power diagram to represent the offset surface.  
While it is challenging to design a “perfect” sampling strategy  
that eliminates all misalignments,  
we propose a refinement technique (Section~\ref{sec:vertex_refinement})  
to effectively suppress them.  
Further discussion is provided in Appendix A.

\textbf{Requirements on $\mathcal{R}$.}
The user-specified radius function~$\mathcal{R}$ must satisfy certain requirements.  
The key condition is that~$\|\nabla \mathcal{R}(p)\| \leq 1$,  
as required by Theorem~\ref{thm:Displacement},  
where only values satisfying this constraint yield a valid angle~$\alpha = \arcsin(\|\nabla \mathcal{R}(p)\|)$.  
This condition is also consistent with reconstruction from the medial axis transform (MAT).  
If~$\|\nabla \mathcal{R}(p)\| > 1$,  
then in the neighborhood of~$p$, there exists another point~$q$ such that the medial ball of~$p$ contains that of~$q$,  
or vice versa—contradicting the MAT property of maximal, non-nested balls.

When the input is a triangle mesh,  
the radius function~$\mathcal{R}$ is typically defined at the vertices.  
We assume that~$\mathcal{R}$ is a linear scalar field within each triangle.  
In practice, users may provide values at only a sparse subset of vertices.  
The remaining values can then be smoothly interpolated by solving the biharmonic equation \( \Delta^2 \mathcal{R} = 0 \) over the mesh surface. Figure~\ref{fig:pyramid_geodesic} shows the interpolated distance fields on the mesh and the corresponding offset surfaces.

\begin{figure}[!t]
    \centering
\includegraphics[width=0.96\linewidth]{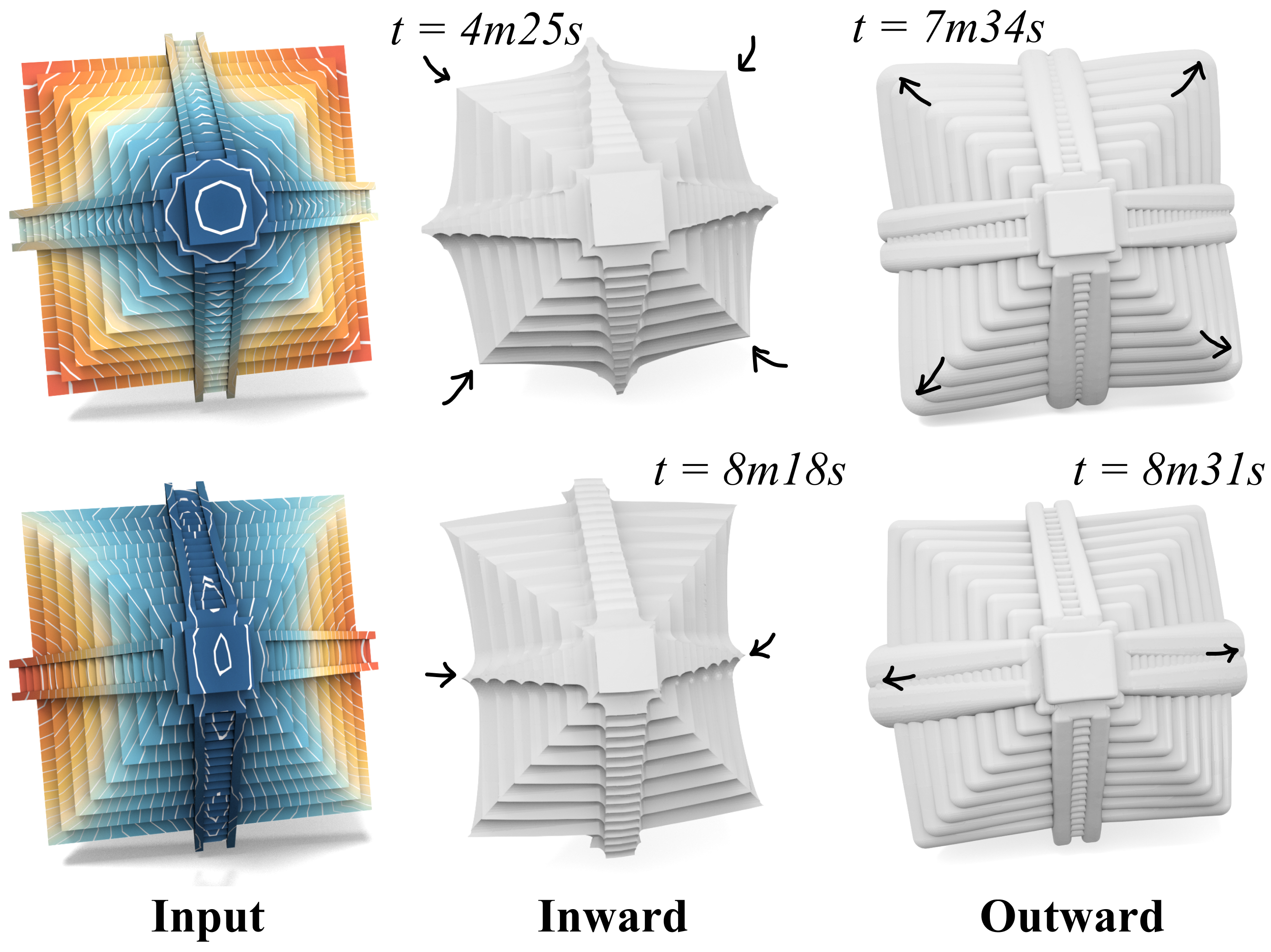}
    \caption{
    Larger radii are specified at four positions on the \textit{top} and two on the \textit{bottom}, and the  distance fields are interpolated over the mesh to produce the final deformations.}
\label{fig:pyramid_geodesic}
\vspace{-5mm}
\end{figure}



\section{Implementation}
\label{sec:sample_methods}

\subsection{Sampling}
In the previous section,  
we assumed that surface normals undergo a smooth transition.  
However, when the input is a triangle mesh,  
surface normals may exhibit abrupt changes across sharp feature lines.  
Therefore, it is necessary to carefully address the sampling problem.

In general, base sample points can be categorized into three types:  
triangle-interior points, edge-type points, and vertex-type points.  
Triangle-interior points are generated using a blue-noise sampling process to mitigate the influence of the input mesh triangulation and better preserve the underlying shape.  
For each triangle-interior point, we generate a unique displaced point based on Theorem~\ref{thm:Displacement}.  
In contrast, for edge-type and vertex-type points,  
we adopt a ``one base point, multiple displaced points'' (1vN) strategy.

\begin{wrapfigure}{r}{0.33\linewidth}
\vspace{1mm}
\includegraphics[width=.99\linewidth]{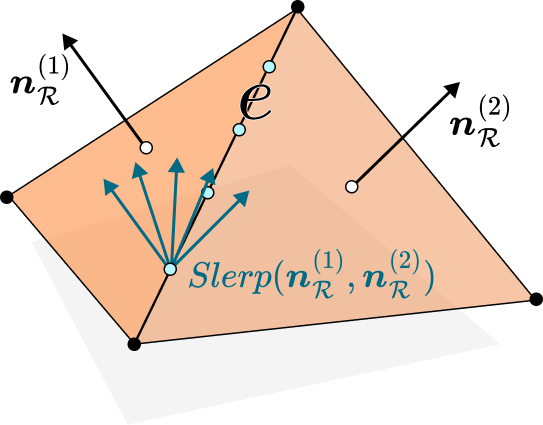}
\vspace{-8mm}
\label{fig:slerp-manifold-edge}
\end{wrapfigure}
\textbf{Edge-Type Base Points.} For the sake of algorithmic generality, 
we do not explicitly detect sharp feature lines.  
Instead, all mesh edges and vertices are treated uniformly.  
The inset figure illustrates the 1vN strategy for generating base points and their corresponding displaced points.  
Consider a manifold edge~$e$ shared by two triangle faces~$f_1$ and~$f_2$.  
For a point on~$e$, the estimated displacement directions  
$\boldsymbol{n}_\mathcal{R}^{(1)}$ and $\boldsymbol{n}_\mathcal{R}^{(2)}$  
may differ depending on whether the estimation is performed within~$f_1$ or~$f_2$.  
Once these directions are computed,  
we generate multiple displacement directions  
by interpolating between them.  
This is efficiently implemented using spherical linear interpolation (Slerp).

\begin{wrapfigure}{r}{0.38\linewidth}
\includegraphics[width=.99\linewidth]{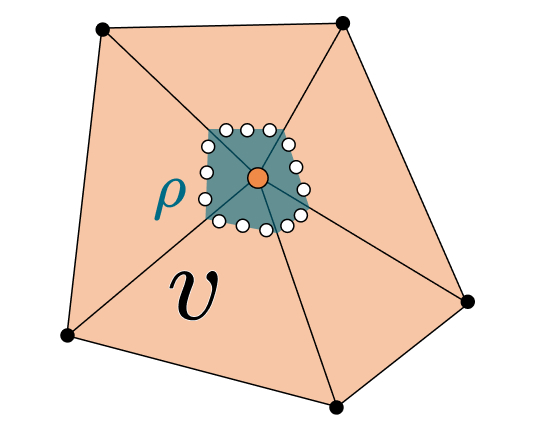}
\vspace{-3mm}
\label{fig:vertex_safe}
\end{wrapfigure}
\textbf{Vertex-Type Base Points.}
For each vertex~\(v\), we place a small sphere centered at~\(v\) 
such that it intersects the base surface along a closed curve (see the inset).
We sample points along this intersection curve and additionally 1vN sample the spherical surface for~\(v\).
Together, these samples form the vertex-type samples.
In practice, we place a discretized spherical sampling pattern centered at~\(v\) and directly sample points on it. Furthermore, we select a set of sample points in the vicinity of~\(v\) and adjust their distances to~\(v\) to match a prescribed radius.
Specifically, the radius of the small sphere is set to $\rho \cdot l$, where $l$ denotes the average edge length.
Any triangle-interior samples that lie entirely inside the sphere are removed, ensuring that displaced points generated from~\(v\) are assigned higher priority.

\textbf{Further Improvements.}
Although treating all edges uniformly and applying the 1vN strategy ensures generality,  
it may introduce unnecessary computational overhead, especially in planar regions.  
Therefore, one may optionally pre-detect sharp feature lines  
by applying a dihedral angle threshold,  
and restrict the use of the 1vN strategy to those regions only.

\begin{wrapfigure}{r}{0.38\linewidth}
\includegraphics[width=.99\linewidth]{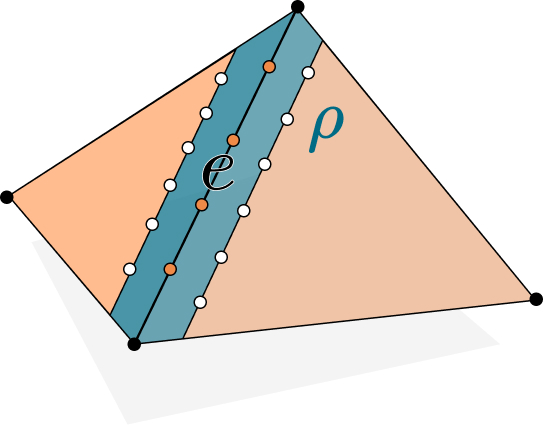}
\vspace{-3mm}
\label{fig:edge_safe}
\end{wrapfigure}
Furthermore, similar to the vertex-based strategy that uses a small enclosing sphere,  
we use a thin cylinder, with a radius of $\rho \cdot l$, to enclose each mesh edge for generating displaced points. 
As shown in the inset figure, the cylinder intersects the base surface along two straight lines.  
Base points are extracted from these lines,  
and corresponding displaced points are then generated.  
It is important to remove any base sample points that lie entirely inside the cylinder  
to ensure a clean, artifact-free offset surface.

In our implementation, the sampling strategy operates on triangle-level geometry and determines how the input mesh corresponds to the offset surfaces.
This gives rise to the following requirements for the input mesh:
\begin{itemize}
\item \textit{Free of self-intersections}: to avoid complicated sharp features introduced by self-intersections rather than by the actual mesh edges.
\item \textit{Good triangle quality}: to accurately detect feature lines and ensure that enclosing spheres and cylinders remain valid even at the triangle level, avoiding issues caused by long or thin triangles.
\item \textit{Watertightness} so that the displaced normal directions are well defined; otherwise, for open surfaces, each triangle must consider offsets on both sides (as in the Medial Axis Transform).
\end{itemize}
In Figure~\ref{fig:feature_detection},
we visualize the base sample points together with their displaced counterparts for both vertex-type and edge-type samples.
Our 1vN sampling strategy effectively handles complex sharp features and produces accurate offset surfaces.
An ablation study on how sampling density affects accuracy is provided in Appendix B.

\begin{figure}[!t]
    \centering
    \includegraphics[width=\linewidth]{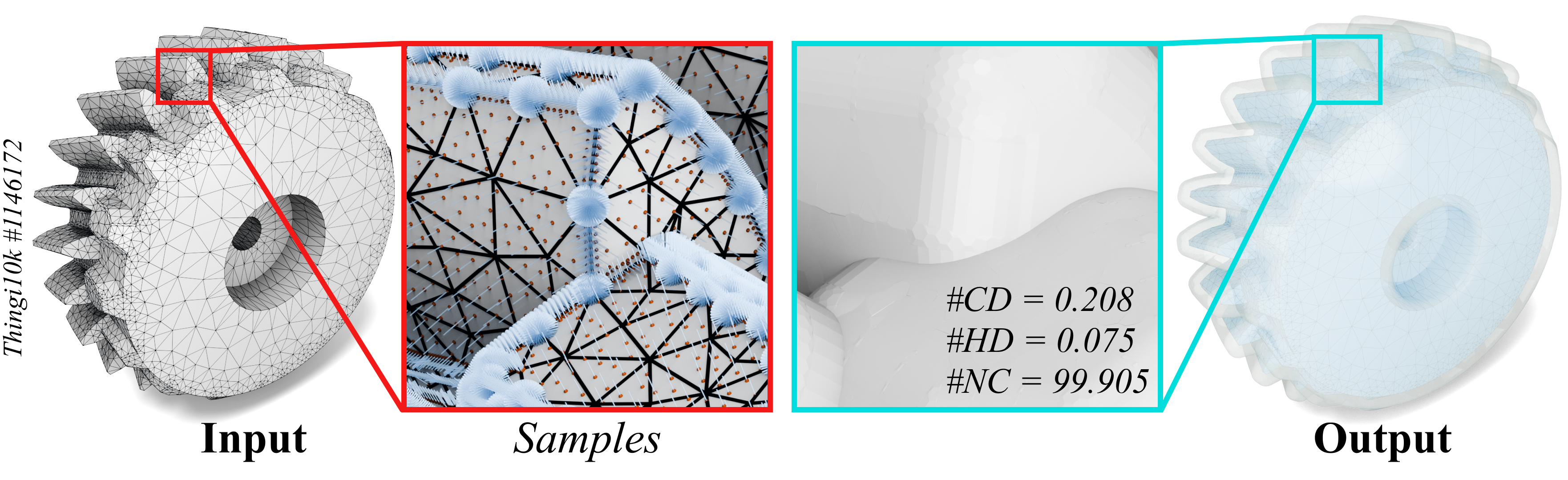}
\caption{
Our 1vN sampling strategy effectively handles complex sharp features, resulting in accurate offset surfaces. Note that Chamfer Distance (CD) is scaled by $10^4$, and both Hausdorff Distance (HD) and Normal Consistency (NC) are scaled by $10^2$.}
    \vspace{-2mm}
    \label{fig:feature_detection}
\end{figure}

\subsection{Misalignment Elimination}
\label{sec:vertex_refinement}

\begin{wrapfigure}{r}{0.52\linewidth}
\vspace{-2mm}
\includegraphics[width=.99\linewidth]{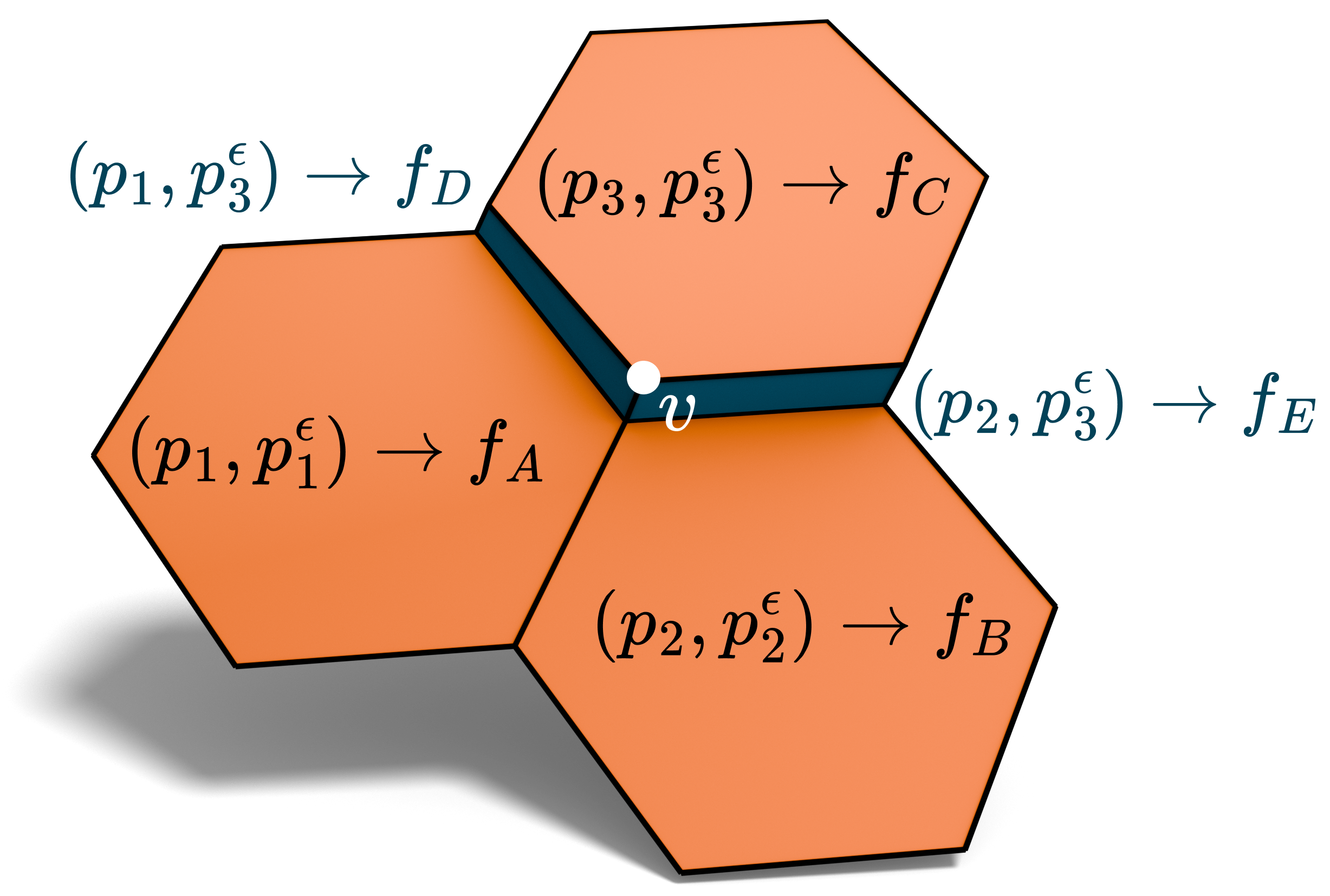}
\vspace{-2mm}
\label{fig:offset d}
\end{wrapfigure} 
As mentioned in the previous section,  
misaligned facets often arise in crust-based approaches.  
Consider the inset figure as an example: each facet is defined by one on-surface point and one off-surface point.  
The vertex~$v$ of the power diagram is incident to three facets: $f_C$, $f_D$, and $f_E$,  
where $f_C$ is formed by $p_3$ and its displaced point,  
$f_D$ by $p_1$ and the displaced point of $p_3$,  
and $f_E$ by $p_2$ and the displaced point of $p_3$.  
This example illustrates that misaligned facets result from competition between one base point  
and the displaced point of another base point.

Generally, each facet incident to the vertex~$v$ is defined by a pair of sites:  
one being a base point and the other a displaced point.  
The two sites may either be a matched pair (e.g., $p_3$ and its displaced point),  
or they may not be paired (e.g., $p_1$ and the displaced point of $p_3$).  
Let $\{p_i\}_{i=1}^k$ denote the set of all base points contributing to a vertex~$v$.  
Ideally, we want~$v$ to lie at a distance~$d_i := \mathcal{R}(p_i)$ from the tangent plane at each~$p_i$.  
Let $\boldsymbol{n}_i^{\text{off}} := \boldsymbol{n}_\mathcal{R}^{\text{off}}(p_i)$ denote the corresponding offset direction.

Formally, we aim to minimize the deviation of the expression:
\begin{equation}
    (v - p_i) \cdot \boldsymbol{n}_i^{\text{off}} - d_i,
\end{equation}
which should be as close to zero as possible.  
This leads naturally to the following least-squares optimization problem:
\begin{equation}
    \min_v \sum_{i=1}^k \left((v - p_i) \cdot \boldsymbol{n}_i^{\text{off}} - d_i\right)^2.
\end{equation}

It can be imagined that when the directions~$\{\boldsymbol{n}_i^{\text{off}}\}$  
are nearly identical, the solution may not be unique.  
This ambiguity can be resolved by encouraging~$v$ to remain close to its original position~$v_0$,  
which leads to a regularized formulation:
\begin{equation}
    \min_v \sum_{i=1}^k \left((v - p_i)\cdot \boldsymbol{n}_i^{\text{off}} - d_i\right)^2 + \lambda \|v - v_0\|^2,
\end{equation}
where $\lambda$ is a small positive constant, typically set to $0.01$.

As long as $\lambda \neq 0$, the solution~$v^*$ exists and is unique, and is given by:
\begin{equation}
    v^* = \boldsymbol{H}^{-1} \left( \lambda v_0 + \boldsymbol{b} \right),
\end{equation}
where
\begin{equation}
    \boldsymbol{H} = \sum_{i=1}^k \boldsymbol{n}_i^{\text{off}} (\boldsymbol{n}_i^{\text{off}})^T + \lambda I,
\end{equation}
and
\begin{equation}
    \boldsymbol{b} = \sum_{i=1}^k (p_i \cdot \boldsymbol{n}_i^{\text{off}} + d_i)\boldsymbol{n}_i^{\text{off}}.
\end{equation}
Similar to the Quadric Error Metrics (QEM) method~\cite{garland1997surface},  
the derivation can also be expressed using homogeneous coordinates.

\begin{figure}[!t]
    \centering
    \includegraphics[width=0.98\linewidth]{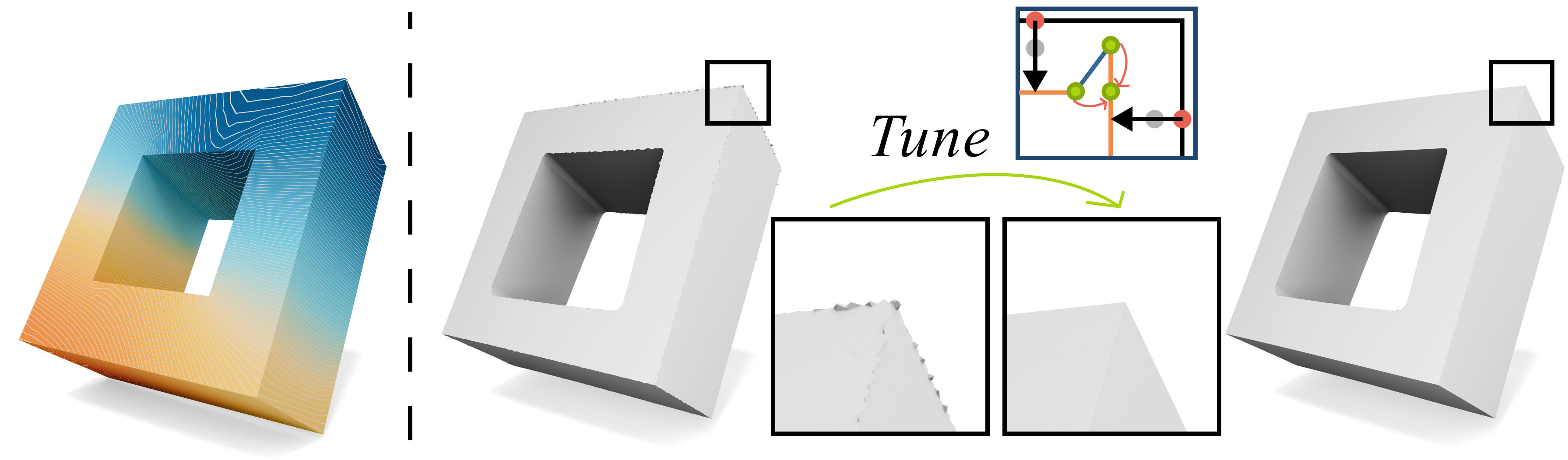}
\caption{
Inward offset of a cube with a square hole, visualized using color-coded varying offset distances.  
Our misalignment elimination strategy effectively optimizes vertex positions while preserving the sharp features of the offset surface.
}
    \vspace{-2mm}
    \label{fig:optim_inside}
\end{figure}
As illustrated in Figure~\ref{fig:optim_inside},  
the vertices are initially misaligned before refinement (see the close-up windows).  
Our misalignment elimination strategy effectively optimizes vertex positions, 
preserving the sharp features of the offset surface.  
However, it should be noted that triangle quality may degrade as a side effect of this optimization.  
Further discussion is provided in the Appendix D.

\section{Evaluation}
\label{sec:evaluation}

\textbf{Platform.}
We implemented our method in C++ using CGAL~\cite{cgal:eb-24a}, employing the \textit{Exact Predicates Exact Constructions Kernel} for robust power diagram computations, accelerated using Intel TBB~\cite{pheatt2008intel}. 
For the misalignment elimination step,  we use Eigen’s LDLT solver~\cite{eigenweb}, parallelized with OpenMP~\cite{dagum1998openmp}. 
We also leverage AABB trees, PQP~\cite{larsen1999fast}, and libigl~\cite{libigl} to support distance and inside-outside queries. 
All experiments were conducted on a machine with an Intel i9-13900K CPU and 64~GB of RAM, running Windows~11.

\textbf{Parameter Setting.}
For triangle mesh inputs, we use blue noise sampling to select 70K points and set the parameters as follows: $\lambda = 0.01$, $\epsilon = 10^{-6}$, $\rho = 5\%$, and a discrete spherical surface with 642 vertices. 
We define a relative offset distance parameter $\delta$ with respect to the bounding box diagonal length $l_{\text{diag}}$, such that the absolute offset distance is given by $d = \delta \cdot l_{\text{diag}}$. 
All input triangle meshes are preprocessed using TetWild~\cite{hu2018tetrahedral, ftetwild} with the edge length parameter set to $l = 0.5$. We denote results without misalignment elimination as ``Ours”, and those with this step as ``Ours+”.


\textbf{Evaluation Metrics.} 
To quantitatively evaluate our results, we adopt three standard reconstruction metrics: \textit{Chamfer Distance} (CD), \textit{Hausdorff Distance} (HD), and \textit{Normal Consistency} (NC). 
CD measures the average squared nearest-neighbor distance, HD captures the maximum squared nearest-neighbor distance, and NC evaluates the alignment of surface normals, computed as the average absolute dot product between corresponding normal vectors. 

For reconstruction-related evaluations, we adopt the standard two-sided formulation, whereas for constant-radius offsets we use the one-sided version that samples the offset surface and projects it onto the original surface.

Specifically for constant offsets, we sample 100K points from the offset surface and compute their projection points onto the original surface, then calculate the projection distances $\mathbf{D}(x)$. 
Accuracy is measured using the normalized distance metric $\left|\mathbf{D}(x) - d\right|$, where $d = \delta \cdot l_{\text{diag}}$.

\subsection{Evaluation on Variable-Radius Offsets}

For variable-radius offsets, the displacement direction may not align with the projection direction, making accurate evaluation challenging. We perform quantitative evaluation using analytical validation, and qualitative evaluation and comparison on triangle mesh inputs.

\begin{figure*}
    \centering
    \includegraphics[width=0.99\linewidth]{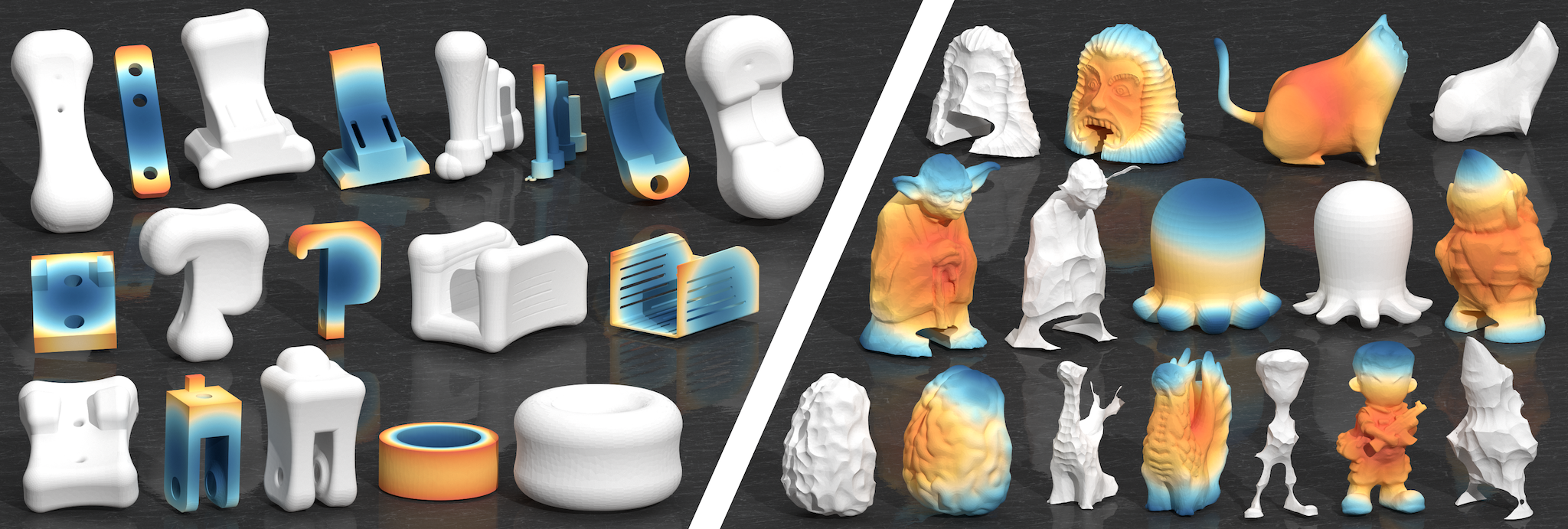}
    \caption{
    A gallery of variable-radius offset surfaces computed by our \textit{OffsetCrust}. 
    The offset distances (visualized using different colors) are generated using the quadratic function $\frac{1}{15}\|\boldsymbol{v}\|^2 + 0.01$, with models normalized to the cube $[-1, 1]^3$. 
    Left: CAD models from the ABC dataset~\cite{Koch_2019_CVPR} and their outward offsets. 
    Right: Freeform models from Thingi10K~\cite{zhou2016thingi10k} and their inward offsets.
    }
    \label{fig:gallery}
\end{figure*}

\begin{figure}[!t]
    \centering
    \includegraphics[width=\linewidth]{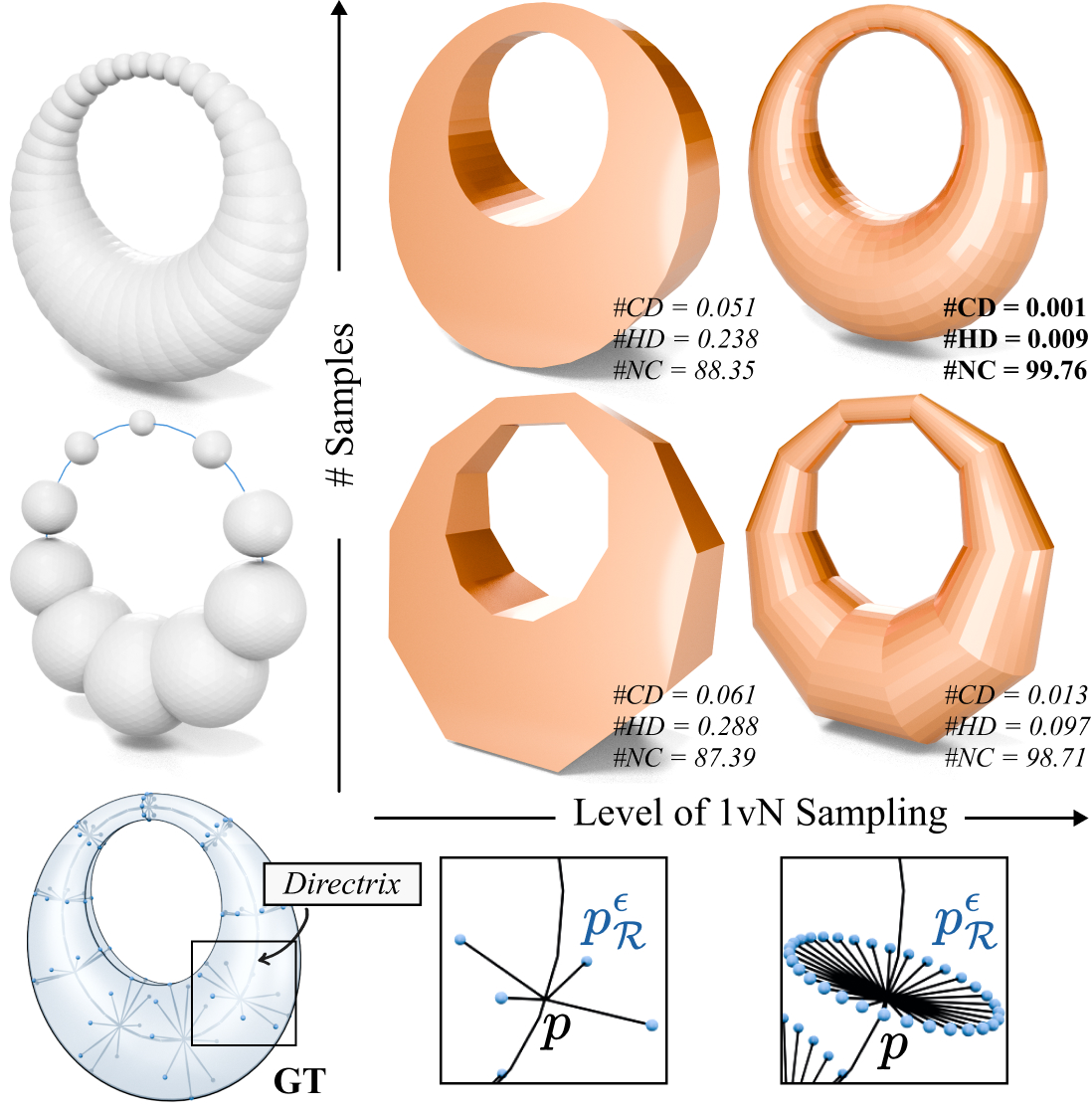}
\caption{
Construction of a Dupin cyclide surface using 1vN sampling along its directrix. Increasing the sampling density yields finer discretization and improves geometric reconstruction accuracy. Moreover, even with a finite number of samples, the results remain topologically correct and preserve piecewise-linear features.
}
    \vspace{-2mm}
    \label{fig:cyclide-mat-1vn}
\end{figure}

\subsubsection{Analytical Validation} 
Handling self-intersections in parametric surface offsets is nontrivial. Therefore, we select the \emph{Dupin cyclide}, a classic shape particularly well-suited for modeling variable-radius offsets.

By definition, the cyclide can be constructed from a parametric curve (called the \textit{directrix}) and a family of spheres with varying radii. To evaluate our method, we sample points along the curve with the corresponding radius to generate the surface, and then compute reconstruction metrics against the analytical ground truth.
Figure~\ref{fig:cyclide-mat-1vn} shows that increasing the sampling density along the curve and using higher-resolution Slerp-based 1vN samples improves surface reconstruction accuracy.
Even with a finite number of samples, our method produces topologically correct surfaces without fractures and oscillation artifacts.




\subsubsection{Comparisons on Challenging Cases}

\begin{figure*}[!t]
    \centering
\includegraphics[width=0.98\linewidth]{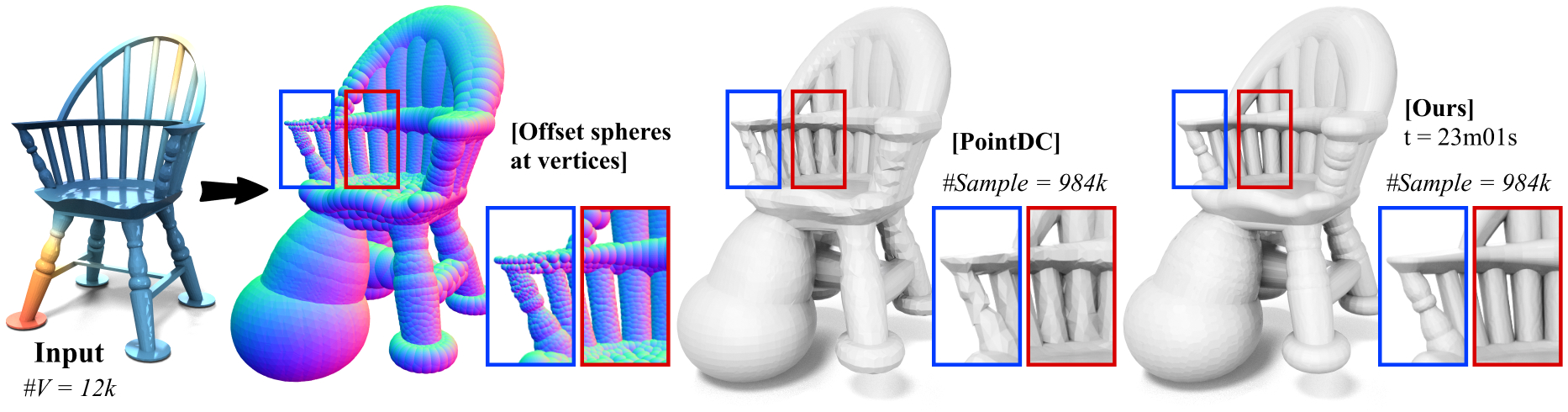}
    \caption{
Offsetting the chair with offset radii defined at the vertices. Two types of features are compared: detailed structures~\textit{(blue)} and narrow gaps~\textit{(red)}. Using only offset spheres centered at the vertices, the reconstruction exhibits oscillation artifacts. 
With the same sample points, \textit{PointDC} produces suboptimal results with prominent artifacts, whereas our method achieves superior quality.
    }
    \label{fig:variable_chair}
\end{figure*}

\begin{figure*}[!t]
    \centering
    \includegraphics[width=.98\linewidth]{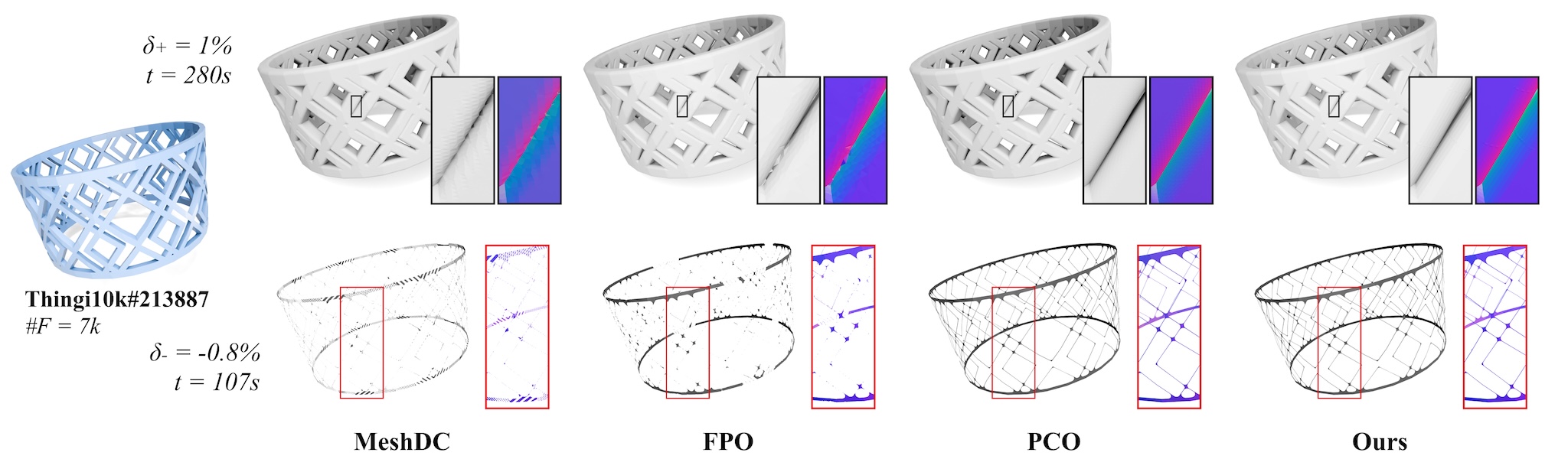}
    \caption{
    Qualitative comparisons of constant-radius offsets. In the outward offset \textit{(top)}, nearby gaps collapse in other methods, whereas both \textit{PCO} and our method preserve the narrow gaps. In the inward offset \textit{(bottom)}, the thin tubes are preserved by \textit{PCO} and our method, while other methods produce disconnected tubes.
    }
    \label{fig:213887}
\end{figure*}

For triangle mesh inputs, we compute variable-radius offsets on both CAD models and freeform models. Figure~\ref{fig:gallery} presents the results of our method, qualitatively demonstrating its effectiveness across different types of triangle meshes.
In particular, narrow gaps and thin tubes represent typical challenging cases in offset computation. We perform comparisons under comparable computational budgets to highlight the advantages of our method in these scenarios.

\textbf{Comparison Methods.}
For variable-radius offset comparisons, we use \textit{PointDC}, which constructs the variable offset distance field directly from Eq.~\ref{eq:offset_distance_field} by computing distances from grid points to the sample points. The offset surface is then extracted using dual contouring (DC)~\cite{ju2002dual}, which preserves sharp features.

For constant-radius offset comparisons, we use \textit{MeshDC}, which constructs the distance field directly from the input mesh to avoid oscillation artifacts, and extracts the offset surface using DC. In addition, we compare against the state-of-the-art sharp-feature-preserving methods: \textit{FPO}~\cite{zint2023feature}, which is based on DC and remeshing techniques, and \textit{PCO}~\cite{wang2024pco}, which relies on tetrahedralization and linear approximation.

\textbf{Narrow Gaps.}
Figure~\ref{fig:variable_chair} shows the outward variable-radius offset of a chair model containing many tubular structures. Due to the incorrect nearest-neighbor searches in distance field computing, \textit{PointDC} fails at narrow gaps. In contrast, our method maintains the correctness of the narrow gaps.
Figure~\ref{fig:213887}~\textit{(top)} compares the outward constant-radius offsets. Both \textit{MeshDC} and \textit{FPO} suffer from gap collapse, whereas \textit{PCO} and our method preserve the narrow gaps.

\textbf{Thin Tubes.}  
Figure~\ref{fig:213887}~\textit{(bottom)} shows the inward constant-radius offsets of the tubes. 
Due to the thin geometry and inaccuracies in the distance field, both \textit{MeshDC} and \textit{FPO} fail to preserve these thin structures, whereas \textit{PCO} and our method successfully maintain them.

\subsection{Robustness and Performance}
\label{sec:eval_robust}

\begin{figure}[!t]
    \centering
    \includegraphics[width=0.98\linewidth]{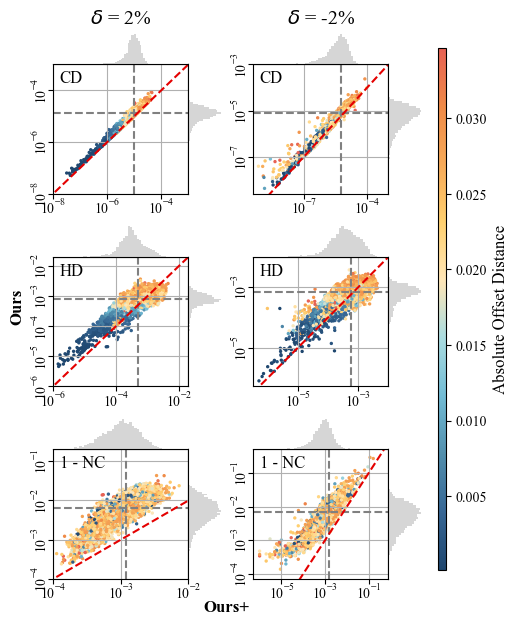}
    \caption{
Quantitative results for constant-radius offsets with $\delta = \pm 2\%$, evaluated on $7,609$ models from Thingi10K, illustrating accuracy and robustness. Colors indicate the absolute offset distance. The gray dashed line marks the average value, while the red line denotes the identity line ($y = x$). Points above the red line correspond to improvements after misalignment elimination.
    }
    \label{fig:thingi10k_res}
\end{figure}

\begin{figure}[!t]
\centering
\includegraphics[width=1\linewidth]{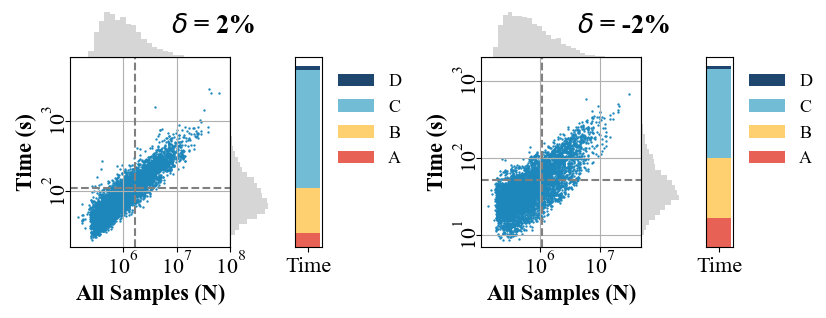}
\caption{
Runtime performance of constant-radius offsets with $\delta=\pm 2\%$ tested on $7,609$ models from Thingi10K. The overall runtime is shown on the \textit{left}, while the time breakdown proportions are illustrated on the \textit{right}.
}
\label{fig:thingi10k_time}
\end{figure}

\begin{figure}[!t]
    \centering
    \includegraphics[width=0.99\linewidth]{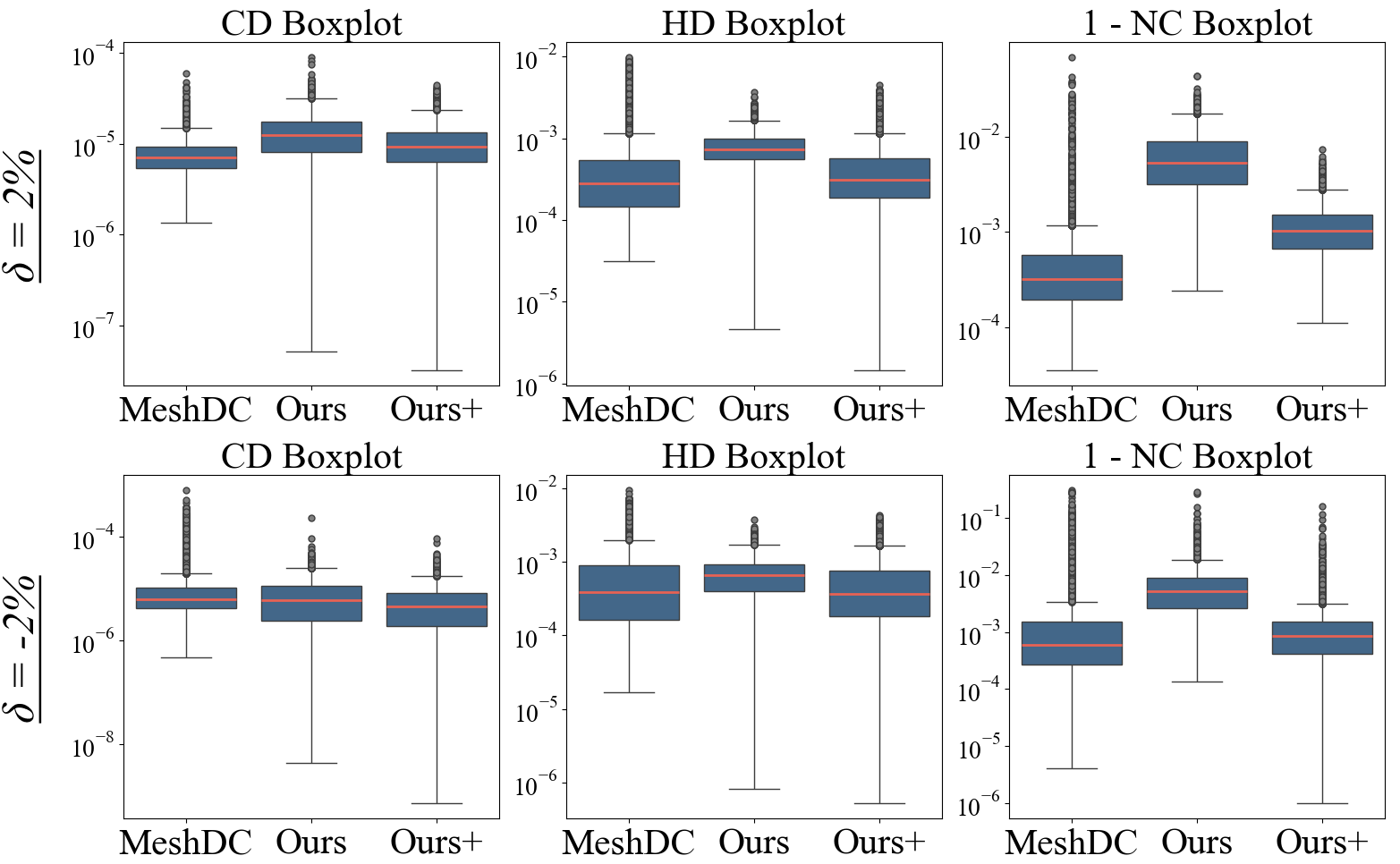}
    \caption{
    Boxplots evaluating \textit{MeshDC} and our method on $7,609$ models from Thingi10K with constant-radius offsets $\delta=\pm 2\%$. The y-axis shows the values of the metrics indicated in the figure title.
    }
    \label{fig:comparison_with_dc}
\end{figure}

To ensure both quality and diversity of the input meshes, we conducted experiments on the Thingi10K dataset~\cite{zhou2016thingi10k} preprocessed by TetWild~\cite{hu2018tetrahedral}. 
We retain $7,609$ valid preprocessed triangle meshes that are watertight, free of self-intersections, and geometrically consistent with their original meshes. All meshes are normalized to the unit cube $[0,1]^3$.
Using default settings, we evaluated our method on relative constant offset distances $\delta = \pm 2\%$. This large-scale evaluation further supports the effectiveness of our default settings across mesh inputs.

\textbf{Quantitative Results and Optimization Impact.}
The quantitative results are presented in Figure~\ref{fig:thingi10k_res}. With relative distances $\delta = \pm 2\%$, various actual offset distances $d$ were tested, consistently demonstrating reasonable accuracy and stability. Before optimization, the method performed well on the distance metrics, CD and HD, indicating that the generated shape closely aligns with the definition of an offset. After optimization, the averages and the distribution of CD and HD show minimal changes, while the NC metric exhibits significant improvement, improving by approximately an order of magnitude.
In general, the initial offset surface achieves reasonable accuracy, but its quality improves significantly after optimization.

\textbf{Performance.}
Figure~\ref{fig:thingi10k_time} shows the runtime performance. The runtime increases as more points are included in the power diagram computation. The average computation time is approximately 100 seconds, with an average of $10^6$ samples. Specifically, our method comprises four major steps: sampling (A), power diagram computation (B), face extraction and adjacency information retrieval (C), and optimization (D). Figure~\ref{fig:thingi10k_time} (right) illustrates the time breakdown for each step.

Specifically, steps A and D consume only a small fraction of the total runtime. Apart from step B, step C is particularly time-consuming. This step is essential for preparing the data for optimization, specifically extracting the polygonal facets and the adjacency information for each vertex within the facets. 
If there is polygonal facet $f$, with each facet having $|v_f|$ vertices, the time complexity for this step is $O\left(\sum_{f} |v_f| \right)$. 
The more remaining facets, the more time is required, which is why inward offsets generally take less time on step C than outward offsets.

\textbf{Overall Quality Assessment}
We use \textit{MeshDC} at a resolution of $300^3$ as a reference and evaluate all valid models with offset distances $\delta = \pm 2\%$. Under these settings, \textit{MeshDC} generates an average of $498K$ faces, while our method generates $511K$ faces, ensuring comparable in mesh resolution.

Figure~\ref{fig:comparison_with_dc} presents boxplots comparing the performance metrics of DC and our method (with and without misalignment elimination). 
For CD and HD, our method is comparable to DC in median and interquartile range. Notably, our method achieves smaller minimum and maximum values, and the range of outliers is also smaller than DC’s, indicating better accuracy and more stable performance.
For NC, although our method improves significantly after eliminating misalignments, it performs better than DC on inward offsets, while the outward offsets do not outperform DC. However, since these differences are within an acceptable range and our method exhibits fewer outliers, the results remain favorable.

\section{Applications}
\label{sec:applications}

In this section, we demonstrate the applicability of \textit{OffsetCrust} across a variety of modeling tasks, including expressive surface design, channel surface modeling, and MAT-based surface reconstruction.

\subsection{Expressive Surface Design} 
\begin{figure*}[!t]
    \centering
    \includegraphics[width=\linewidth]{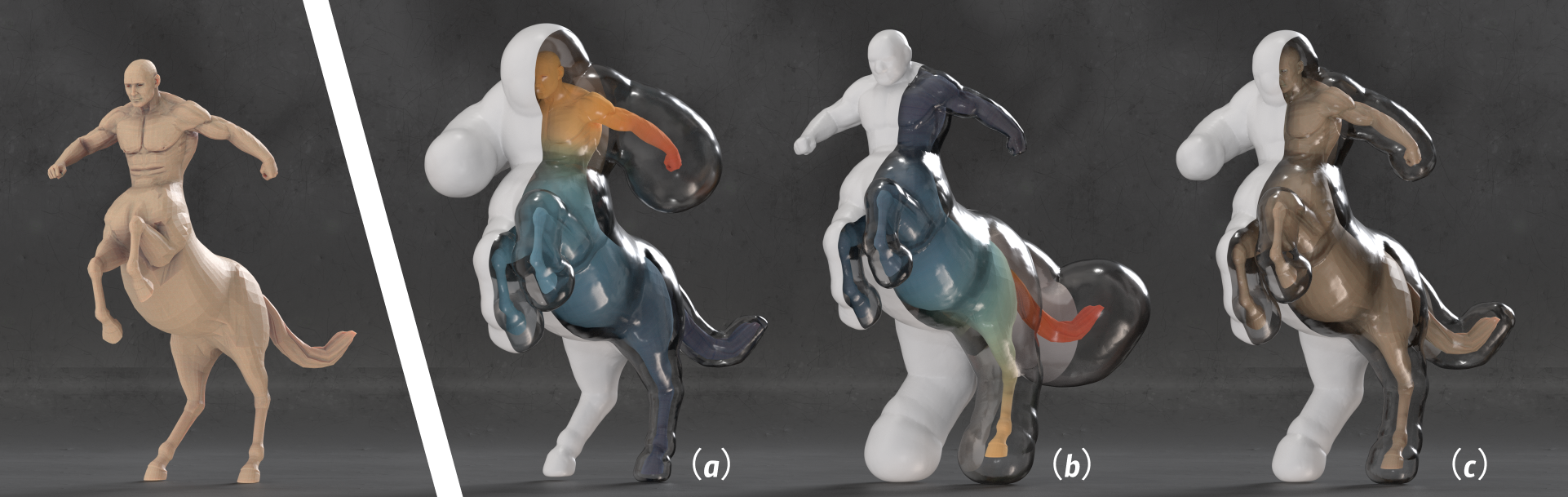}
  \vspace{-5mm}
  \caption{By assigning varying offset distances (visualized using different colors) to the vertices of the input Centaur model, our method generates a set of outward variable-radius offset surfaces: (a) the front part is enlarged significantly more than the rear, (b) the rear part is enlarged significantly more than the front, and (c) a constant-radius offset is applied uniformly across the model.
  }
  \label{fig:teaser}
\end{figure*}

Our method supports variable-radius offsets on triangle meshes, allowing offset distances to be specified per vertex or continuously across the surface. 
By assigning distances to selected vertices and interpolating them via the biharmonic equation ($\Delta^2 \mathcal{R} = 0$), our approach enables flexible and intuitive modeling of diverse, expressive surface variations, making it well-suited for customized freeform surface design. 
For example, as illustrated in Figure~\ref{fig:teaser}, we can compute different types of offset surfaces on the Centaur model to create a well-fitted armor.

\subsection{Offset-Based Channel Surface Modeling} 
\begin{figure*}[!t]
    \centering
    \includegraphics[width=\linewidth]{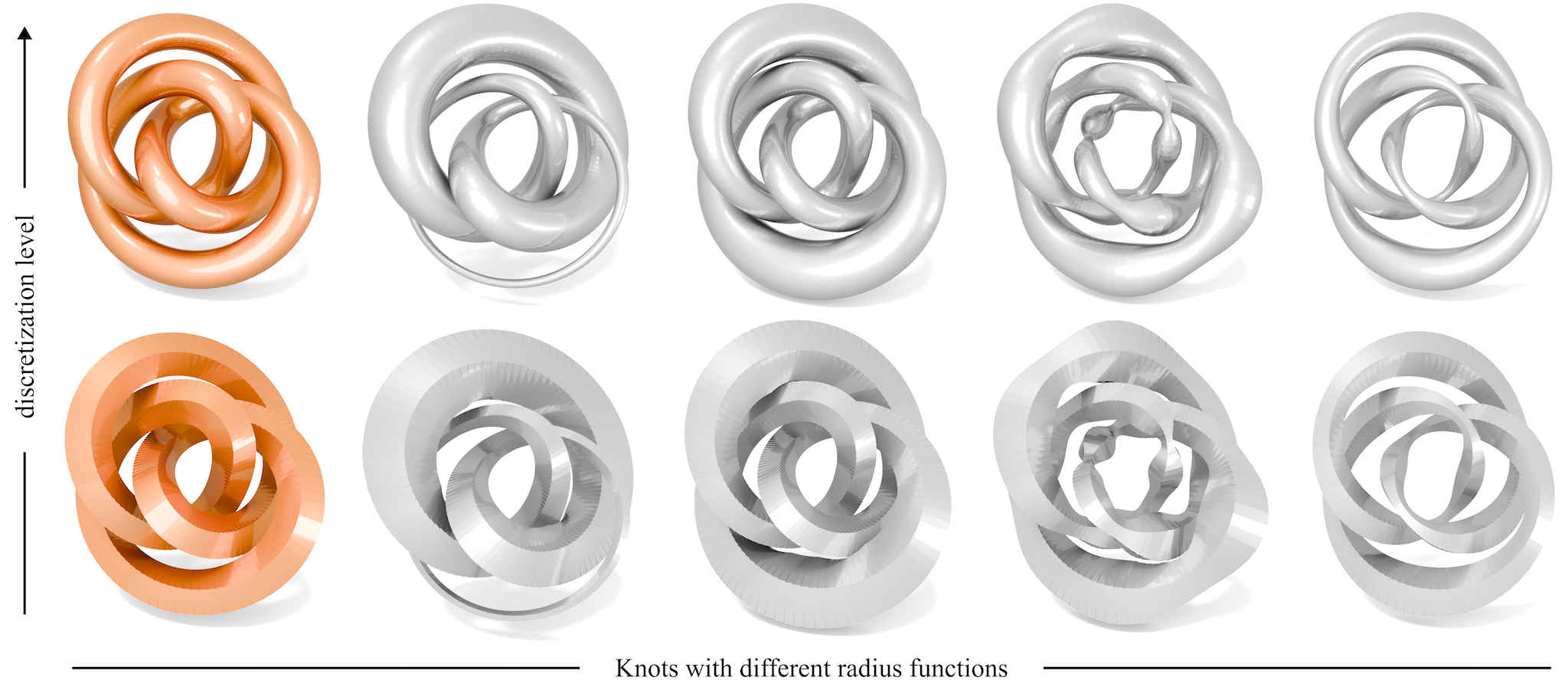}
\caption{
Construction of diverse knots from a parameterized central curve, defined as $x = (2 + \cos(2u)) \cos(3u)$, $y = (2 + \cos(2u)) \sin(3u)$, and $z = \sin(4u)$, with 300 sample points. 
From left to right, the radius functions vary as follows: 
(a) $\mathcal{R}(u) = 0.45$ (constant radius), 
(b) $\mathcal{R}(u) = 0.4 + 0.3\cos(u)$, 
(c) $\mathcal{R}(u) = 0.4 + 0.3 e^{-5 \sin^2(u)}$, 
(d) $\mathcal{R}(u) = 0.4 + 0.15 \cos(2u) + 0.1 \sin(13u + \cos(u)) + 0.05\sin(20u)$, 
(e) $\mathcal{R}(u) = 0.3 + 0.1\cos(2u) + 0.1\cos(8u)$. 
\textit{Top}: limiting displacement directions to 150 possibilities. 
\textit{Bottom}: limiting displacement directions to only 5 possibilities.
}

    \vspace{2mm}
    \label{fig:knot radii}
\end{figure*}

Channel surfaces are formed as the envelopes of spheres whose centers lie along a curve known as the \emph{directrix}~\cite{peternell1997computing}. Using our \emph{OffsetCrust} method, we begin with a user-specified curve and an associated smooth radius function. The method provides flexibility in defining non-uniform radius profiles, enabling the modeling of complex shapes.

In Figure~\ref{fig:knot radii}, we generate diverse knot models by varying the radius function while keeping the directrix fixed. Furthermore, we can intentionally reduce the number of displacement directions: the top row of Figure~\ref{fig:knot radii} limits displacement directions to 150, while the bottom row further reduces them to just 5, resulting in a faceted offset surface.

\begin{figure*}[!t]
    \centering
    \includegraphics[width=.99\linewidth]{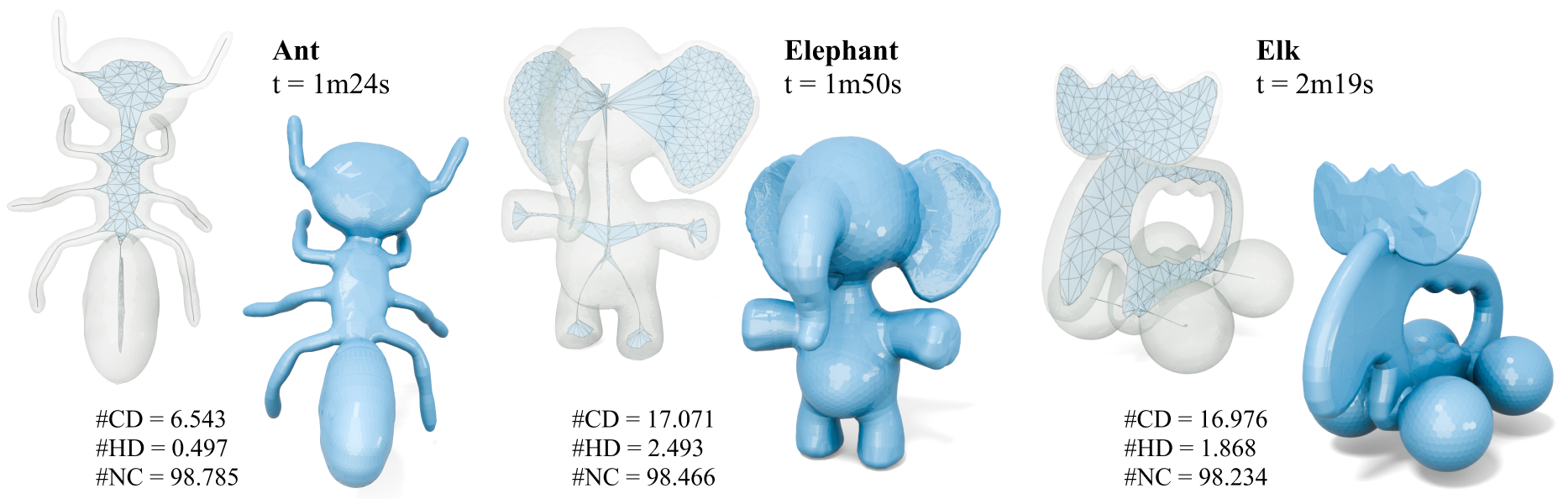}
\caption{
Surface reconstruction from medial axis transforms (MATs) generated using variable-radius offsets computed by our method. The MATs are produced by Q-MAT~\cite{li2015qmat}, and the reconstructed surfaces faithfully preserve the geometry of the original shapes. In these examples, we use 5K blue noise samples for on-surface sampling.
}
    \vspace{-2mm}
    \label{fig:mesh_mat}
\end{figure*}
\subsection{Surface Reconstruction from Medial Axis Transform} 

As a shape descriptor~\cite{wang2022computing, wang2024mattopo, li2015qmat}, the medial axis transform (MAT) offers a compact representation of a shape’s topology and structure. In the discrete setting, the MAT surface consists of vertices, edges, and faces, with an associated radius defined at each vertex.
Despite significant progress in computing MATs, recovering the original boundary surface from a given MAT remains a challenging problem.

In fact, this task can be formulated as a variable-radius offsetting problem, where the base surface is the MAT itself, and the radius function naturally satisfies $\|\nabla\mathcal{R}\| < 1$. Figure~\ref{fig:mesh_mat} presents surface reconstruction results, where the input MATs are generated using Q-MAT~\cite{li2015qmat}. 
We evaluate the reconstruction quality by computing the two-sided distance between the original boundary surface and the reconstructed surface. Based on the CD, HD, and NC scores shown in the figure, it is evident that our \emph{OffsetCrust} method accurately recovers the geometry of the original boundary surface.

\section{Discussion and Limitations}

In this paper, we presented \emph{OffsetCrust}, a crust-based method for explicitly computing variable-radius offset surfaces. Given the radius function~$\mathcal{R}$, our formulation leverages power diagrams to robustly approximate the offset geometry by carefully sampling base points and generating their corresponding off-surface points, displaced along~$\mathcal{R}$-dependent directions. This approach generalizes smoothly from the constant-radius case—where displacement directions align with surface normals—to more complex, variable-radius scenarios.

To address the misalignment issues commonly encountered in crust-based methods, we introduced a lightweight fine-tuning procedure that significantly improves geometric fidelity. Through extensive experiments, we validated the effectiveness and efficiency of our method and demonstrated its applicability in a variety of modeling tasks, including MAT-based surface reconstruction.
Despite these strengths, our method has certain limitations. 

\begin{figure}[!t]
    \centering
    \includegraphics[width=.94\linewidth]{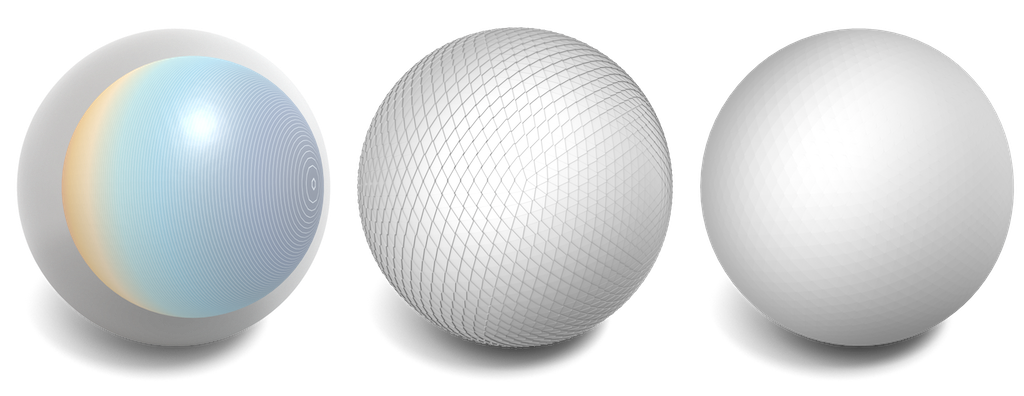}
    \caption{
    Variable-radius offset of a discrete sphere \textit{(left)}. Without rotating the normals, unresolvable large misalignments occur \textit{(middle)}. After applying the rotated normals by Theorem~\ref{thm:Displacement}, these misalignments are eliminated \textit{(right)}.
    }
    \label{fig:sphere_var-offset}
    \vspace{-3mm}
\end{figure}

First, due to the domain restriction of $\arcsin(\cdot)$, we must ensure that $\|\nabla \mathcal{R}(p)\| \leq 1$. 
If this constraint is violated, the displacement directions cannot be uniquely determined, which leads to unresolvable misalignments (see Figure~\ref{fig:sphere_var-offset}). 
As a result, this constraint limits the types of variable-radius inputs that can be handled, particularly those with extremely rapid variations in the offset distance.

Second, the proposed sampling strategies are heuristically designed to capture face-, vertex-, and edge-type features on triangle meshes, with the number of samples determined empirically. 
To ensure sufficiently uniform coverage, we employ blue-noise sampling, which may result in dense output meshes even for simple inputs. 
However, as shown in Section~\ref{sec:eval_robust}, the default settings work effectively for preprocessed input meshes from Thingi10K. 
Further, meshes with challenging geometries can be reliably remeshed using TetWild~\cite{ftetwild} to produce suitable inputs for our method.

Third, the output mesh is dense and of low quality (see Appendix E), even though the geometric shape is correct. This is because our method directly adopts the polygonal facets of the power diagram as offset surface patches, and each pair of samples corresponds to a polygon facet. The triangular faces are generated via ear clipping~\cite{eberly2008triangulation}, followed by post-processing to remove duplicate vertices and degenerate faces~\cite{libigl}. Similar low-quality meshes can be observed in naïve Voronoi-based MAT computations~\cite{li2015qmat,amenta2001power}. 
However, since the current applications of variable-radius offsets are primarily focused on surface design, having a mesh with correct geometry but suboptimal triangulation is acceptable. Post-processing techniques such as remeshing (e.g., TetWild~\cite{ftetwild}) can still be applied if higher mesh quality is required.

\section*{Acknowledgment} 

The authors thank the anonymous reviewers for their insightful comments and suggestions. 
This work was supported by the Joint Funds of the National Natural Science Foundation of China (U23A20312), the Natural Science Foundation of Shandong Province (ZR2025MS986), and the Key Research and Development Program of Shandong Province (2024TSGC0118).

\bibliographystyle{IEEEtran}
\bibliography{references}

\appendices
\section{Analysis of Misaligned Facets}
\label{sec:appendix misaligned_facets}

Based on Section 3.2, we can filter the offset facets. However, due to the discontinuities in real sampling, some additional facets are also included to preserve watertightness: 
\begin{equation} 
d_{pow}(x,p_i) = d_{pow}(x,p_j^\epsilon), \quad i \neq j. \end{equation} 
Meanwhile, we do not select the facets that satisfy: 
\begin{align} 
d_{pow}(x,p_i) &= d_{pow}(x,p_j),\quad i \neq j, \\
\text{or} \quad d_{pow}(x,p_i^\epsilon) &= d_{pow}(x,p_j^\epsilon),\quad i \neq j. \end{align}

To completely eliminate the misaligned facets, we require that for all $x \in \{x \mid d_{pow}(x,p_i) = d_{pow}(x,p_j^\epsilon)\}$, the following must hold: 
\begin{align} d_{pow}(x,p_j) &\leq d_{pow}(x,p_i) = d_{pow}(x,p_j^\epsilon), \quad i \neq j, \\
\text{or} \quad d_{pow}(x,p_i^\epsilon) &\leq d_{pow}(x,p_j^\epsilon) = d_{pow}(x,p_i), \quad i \neq j. 
\end{align}
Conversely, misaligned facets \textbf{exist} when: 
\begin{align} d_{pow}(x,p_j) &> d_{pow}(x,p_i) = d_{pow}(x,p_j^\epsilon), \quad i \neq j, \\
\text{and} \quad d_{pow}(x,p_i^\epsilon) &> d_{pow}(x,p_j^\epsilon) = d_{pow}(x,p_i), \quad i \neq j. 
\end{align}

For any four random sites, to eliminate both sides of misaligned facets, the following condition must be satisfied: \begin{equation} d_{pow}(x, p_i) = d_{pow}(x, p_i^\epsilon) = d_{pow}(x, p_j) = d_{pow}(x, p_j^\epsilon). \end{equation}
Figure~\ref{fig:illu_mis_pd} demonstrates the relationship between the regions for four random sites. Misaligned facets are almost inevitable in any sampling configuration when implementing \textit{OffsetCrust}. 

Figure~\ref{fig:pd_epsilon_misaligned} illustrates the configuration of sites used in \textit{OffsetCrust}. While $0 < \epsilon < \mathcal{R}(p)$ is theoretically effective, as analyzed in Section 3, selecting a small $\epsilon > 0$ improves stability when addressing misaligned facets. This is because the $\epsilon$-balls have radii comparable to the offset balls, making their interactions with other sites more similar to prevent uncontrollable misalignments (Figure~\ref{fig:illu_mis_eps}). As a result, local adjacency is better preserved, which benefits vertex optimization in Section 4.2.

\begin{figure}[!t]
    \centering
    \includegraphics[width=0.98\linewidth]{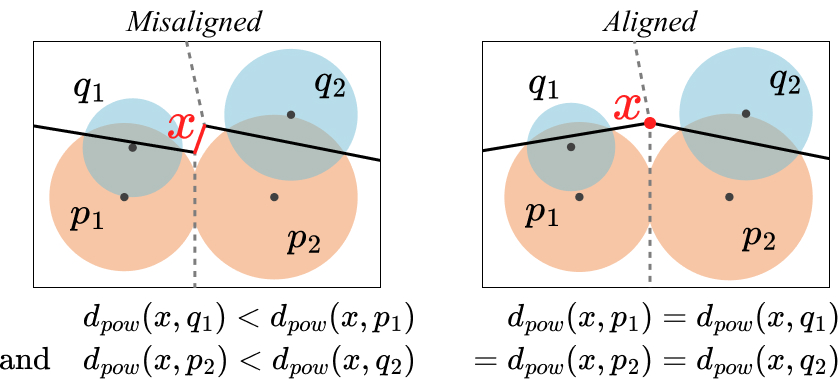}
    \caption{Power diagram of four random sites. }
    \label{fig:illu_mis_pd}
\end{figure}

\begin{figure}[!t]
    \centering
    \includegraphics[width=0.98\linewidth]{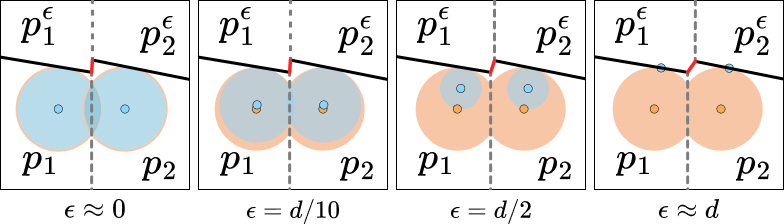}
    \caption{
    Configurations of sites in \emph{OffsetCrust}. With a larger $\epsilon$, the misaligned facets become more oblique relative to the sites, which can easily lead to inconsistent regions between $p_i$ and $p_i^\epsilon$.
}
    \label{fig:pd_epsilon_misaligned}
\end{figure}

\begin{figure}[!t]
    \centering
    \includegraphics[width=0.98\linewidth]{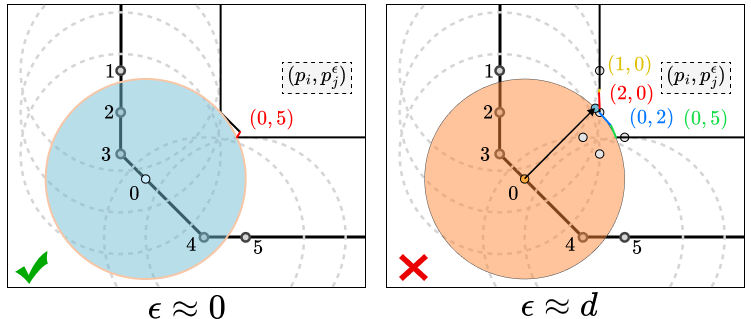}
    \caption{
    With smaller $\epsilon$, the misalignments are more localized and thus easier to optimize. In contrast, larger $\epsilon$ leads to more complex misalignment adjacencies, making optimization more difficult.}
    \label{fig:illu_mis_eps}
\end{figure}

\section{Ablation Study}
\label{sec:ablation_study}


\begin{table*}[!t]\fontsize{6.5pt}{6.5pt}\selectfont
\centering
\setlength{\tabcolsep}{1.8pt}
\caption{Ablation studies on the smooth kitten model to examine the effects of the number of blue noise samples, and on the sharp block model to test the size of the protected sharp regions and the level of sphere discretization. 
The gray-filled \colorbox[RGB]{237, 237, 237}{cells} show scores of optimized results. Among these, the best scores are emphasized in \underline{\textbf{bold}}, while the second-best scores are highlighted in \textbf{bold}. For the ablation study of block model, \textcolor[rgb]{0, 0, 0.6078}{blue} data indicate significant changes in the scores for each group of settings.
\label{tb:ablation_study}}
\vspace{0mm}
\renewcommand{\arraystretch}{1.6}
\begin{tabular}{cc|cccccccc||cccccccccccc}
\Xhline{2.3\arrayrulewidth}
\multicolumn{2}{c|}{\textit{Ours / Ours+}}                                  & \multicolumn{8}{c||}{\textbf{Kitten (Smooth)}}                                                                                                                                                                                               & \multicolumn{12}{c}{\textbf{Block (Sharp   Features)}}                                                                                                                                                                                                                                                                                                                                                                                                                                                                                                                                     \\ \hline
\multicolumn{2}{c|}{\cellcolor[HTML]{EDEDED}}                                    & \multicolumn{2}{c}{$B_1$}                                 & \multicolumn{2}{c}{$B_2$}                                 & \multicolumn{2}{c}{$B_3$}                                  & \multicolumn{2}{c||}{$B_4$}                                 & \multicolumn{2}{c}{$R_1$}                                                              & \multicolumn{2}{c}{$R_2$}                                                              & \multicolumn{2}{c|}{$R_3$}                                                                                                & \multicolumn{2}{c|}{$R_4$, $S_2$}                                                  & \multicolumn{2}{c}{$S_3$}                                                                             & \multicolumn{2}{c}{$S_1$}                                                              \\
\multicolumn{2}{c|}{\multirow{-2}{*}{\cellcolor[HTML]{EDEDED}\textit{Settings}}} & \multicolumn{2}{c}{\cellcolor[HTML]{EDEDED}\#blue = 35K} & \multicolumn{2}{c}{\cellcolor[HTML]{EDEDED}\#blue = 70K} & \multicolumn{2}{c}{\cellcolor[HTML]{EDEDED}\#blue = 140K} & \multicolumn{2}{c||}{\cellcolor[HTML]{EDEDED}\#blue = 550K} & \multicolumn{2}{c}{\cellcolor[HTML]{EDEDED}NaN}                   & \multicolumn{2}{c}{\cellcolor[HTML]{EDEDED}$\rho=70\%$}                  & \multicolumn{2}{c|}{\cellcolor[HTML]{EDEDED}$\rho=30\%$}                                                    & \multicolumn{2}{c|}{\cellcolor[HTML]{EDEDED}$\rho=5\%$, \#SV = 642} & \multicolumn{2}{c}{\cellcolor[HTML]{EDEDED}\#SV = 2562}                                              & \multicolumn{2}{c}{\cellcolor[HTML]{EDEDED}\#SV = 162}                                \\ \hline
\multicolumn{1}{c|}{}                                     & CD $(\times 10^4) \downarrow$             & 0.463           & \cellcolor[HTML]{EDEDED}0.442          & 0.297           & \cellcolor[HTML]{EDEDED}0.277          & 0.188       & \cellcolor[HTML]{EDEDED}\textbf{0.161}      & 0.046    & \cellcolor[HTML]{EDEDED}{\ul \textbf{0.044}}   & 4.276                         & \cellcolor[HTML]{EDEDED}3.986                         & 4.225                         & \cellcolor[HTML]{EDEDED}3.922                         & 4.190                         & \multicolumn{1}{c|}{\cellcolor[HTML]{EDEDED}\textbf{3.866}}                              & 4.050     & \multicolumn{1}{c|}{\cellcolor[HTML]{EDEDED}3.901}                   & {\color[HTML]{00009B} 2.475}  & \cellcolor[HTML]{EDEDED}{\color[HTML]{00009B} {\ul \textbf{2.170}}}  & {\color[HTML]{00009B} 11.365} & \cellcolor[HTML]{EDEDED}{\color[HTML]{00009B} 11.616} \\
\multicolumn{1}{c|}{}                                     & HD $(\times 10^2) \downarrow$              & 0.452           & \cellcolor[HTML]{EDEDED}0.394          & 0.332           & \cellcolor[HTML]{EDEDED}0.349          & 0.277       & \cellcolor[HTML]{EDEDED}\textbf{0.256}      & 0.078    & \cellcolor[HTML]{EDEDED}{\ul \textbf{0.054}}   & 2.751                         & \cellcolor[HTML]{EDEDED}1.150                         & 2.752                         & \cellcolor[HTML]{EDEDED}\textbf{1.061}                & 2.539                         & \multicolumn{1}{c|}{\cellcolor[HTML]{EDEDED}1.081}                                       & 3.449     & \multicolumn{1}{c|}{\cellcolor[HTML]{EDEDED}1.091}                   & {\color[HTML]{00009B} 3.358}  & \cellcolor[HTML]{EDEDED}{\color[HTML]{00009B} {\ul \textbf{1.040}}}  & {\color[HTML]{00009B} 3.134}  & \cellcolor[HTML]{EDEDED}{\color[HTML]{00009B} 3.053}  \\
\multicolumn{1}{c|}{}                                     & NC $(\times 10^2) \uparrow$               & 98.237          & \cellcolor[HTML]{EDEDED}99.654         & 98.538          & \cellcolor[HTML]{EDEDED}99.699         & 98.658      & \cellcolor[HTML]{EDEDED}\textbf{99.725}     & 99.719   & \cellcolor[HTML]{EDEDED}{\ul \textbf{99.930}}  & 99.258                        & \cellcolor[HTML]{EDEDED}99.902                        & 99.416                        & \cellcolor[HTML]{EDEDED}99.907                        & 99.446                        & \multicolumn{1}{c|}{\cellcolor[HTML]{EDEDED}99.908}                                      & 99.384    & \multicolumn{1}{c|}{\cellcolor[HTML]{EDEDED}\textbf{99.914}}         & {\color[HTML]{00009B} 99.471} & \cellcolor[HTML]{EDEDED}{\color[HTML]{00009B} {\ul \textbf{99.921}}} & {\color[HTML]{00009B} 99.187} & \cellcolor[HTML]{EDEDED}{\color[HTML]{00009B} 99.841} \\
\multicolumn{1}{c|}{\multirow{-4}{*}{-4\%}}               & Time (s)             & \multicolumn{2}{c}{44.535}                               & \multicolumn{2}{c}{68.653}                               & \multicolumn{2}{c}{107.935}                               & \multicolumn{2}{c||}{254.460}                              & \multicolumn{2}{c}{20.805}                                                            & \multicolumn{2}{c}{25.697}                                                            & \multicolumn{2}{c|}{27.910}                                                                                              & \multicolumn{2}{c|}{47.452}                                                      & \multicolumn{2}{c}{89.990}                                                                           & \multicolumn{2}{c}{40.651}                                                            \\ \hline
\multicolumn{1}{c|}{}                                     & CD $(\times 10^4) \downarrow$             & 0.600           & \cellcolor[HTML]{EDEDED}0.569          & 0.324           & \cellcolor[HTML]{EDEDED}0.303          & 0.186       & \cellcolor[HTML]{EDEDED}\textbf{0.162}      & 0.029    & \cellcolor[HTML]{EDEDED}{\ul \textbf{0.028}}   & 7.604                         & \cellcolor[HTML]{EDEDED}\textbf{6.987}                & 7.680                         & \cellcolor[HTML]{EDEDED}7.034                         & 7.563                         & \multicolumn{1}{c|}{\cellcolor[HTML]{EDEDED}7.139}                                       & 7.235     & \multicolumn{1}{c|}{\cellcolor[HTML]{EDEDED}7.148}                   & {\color[HTML]{00009B} 2.640}  & \cellcolor[HTML]{EDEDED}{\color[HTML]{00009B} {\ul \textbf{2.436}}}  & {\color[HTML]{00009B} 27.176} & \cellcolor[HTML]{EDEDED}{\color[HTML]{00009B} 28.004} \\
\multicolumn{1}{c|}{}                                     & HD $(\times 10^2) \downarrow$              & 0.320           & \cellcolor[HTML]{EDEDED}0.310          & 0.224           & \cellcolor[HTML]{EDEDED}0.287          & 0.145       & \cellcolor[HTML]{EDEDED}\textbf{0.080}      & 0.035    & \cellcolor[HTML]{EDEDED}{\ul \textbf{0.033}}   & 1.665                         & \cellcolor[HTML]{EDEDED}0.799                         & 2.719                         & \cellcolor[HTML]{EDEDED}0.872                         & 1.660                         & \multicolumn{1}{c|}{\cellcolor[HTML]{EDEDED}\textbf{0.783}}                              & 1.932     & \multicolumn{1}{c|}{\cellcolor[HTML]{EDEDED}0.908}                   & {\color[HTML]{00009B} 1.800}  & \cellcolor[HTML]{EDEDED}{\color[HTML]{00009B} {\ul \textbf{0.773}}}  & {\color[HTML]{00009B} 3.276}  & \cellcolor[HTML]{EDEDED}{\color[HTML]{00009B} 2.970}  \\
\multicolumn{1}{c|}{}                                     & NC $(\times 10^2) \uparrow$               & 98.480          & \cellcolor[HTML]{EDEDED}99.705         & 98.661          & \cellcolor[HTML]{EDEDED}99.737         & 98.823      & \cellcolor[HTML]{EDEDED}\textbf{99.763}     & 99.828   & \cellcolor[HTML]{EDEDED}{\ul \textbf{99.962}}  & 99.359                        & \cellcolor[HTML]{EDEDED}99.869                        & 99.517                        & \cellcolor[HTML]{EDEDED}99.877                        & 99.582                        & \multicolumn{1}{c|}{\cellcolor[HTML]{EDEDED}99.877}                                      & 99.530    & \multicolumn{1}{c|}{\cellcolor[HTML]{EDEDED}\textbf{99.877}}         & {\color[HTML]{00009B} 99.693} & \cellcolor[HTML]{EDEDED}{\color[HTML]{00009B} {\ul \textbf{99.902}}} & {\color[HTML]{00009B} 98.999} & \cellcolor[HTML]{EDEDED}{\color[HTML]{00009B} 99.657} \\
\multicolumn{1}{c|}{\multirow{-4}{*}{4\%}}                & Time (s)             & \multicolumn{2}{c}{34.090}                               & \multicolumn{2}{c}{54.574}                               & \multicolumn{2}{c}{105.528}                               & \multicolumn{2}{c||}{275.462}                              & \multicolumn{2}{c}{46.316}                                                            & \multicolumn{2}{c}{62.449}                                                            & \multicolumn{2}{c|}{63.389}                                                                                              & \multicolumn{2}{c|}{141.155}                                                     & \multicolumn{2}{c}{209.596}                                                                          & \multicolumn{2}{c}{86.397}                                                            \\ \hline
\multicolumn{1}{c|}{}                                     & CD $(\times 10^4) \downarrow$             & 0.261           & \cellcolor[HTML]{EDEDED}0.234          & 0.140           & \cellcolor[HTML]{EDEDED}0.125          & 0.076       & \cellcolor[HTML]{EDEDED}\textbf{0.068}      & 0.004    & \cellcolor[HTML]{EDEDED}{\ul \textbf{0.005}}   & {\color[HTML]{00009B} 4.312}  & \cellcolor[HTML]{EDEDED}{\color[HTML]{00009B} 4.098}  & {\color[HTML]{00009B} 2.237}  & \cellcolor[HTML]{EDEDED}{\color[HTML]{00009B} 2.054}  & {\color[HTML]{00009B} 1.149}  & \multicolumn{1}{c|}{\cellcolor[HTML]{EDEDED}{\color[HTML]{00009B} 0.709}}                & 1.409     & \multicolumn{1}{c|}{\cellcolor[HTML]{EDEDED}\textbf{0.664}}          & 1.349                         & \cellcolor[HTML]{EDEDED}{\ul \textbf{0.644}}                         & 1.419                         & \cellcolor[HTML]{EDEDED}0.740                         \\
\multicolumn{1}{c|}{}                                     & HD $(\times 10^2) \downarrow$              & 0.129           & \cellcolor[HTML]{EDEDED}0.085          & 0.100           & \cellcolor[HTML]{EDEDED}0.074          & 0.053       & \cellcolor[HTML]{EDEDED}\textbf{0.049}      & 0.006    & \cellcolor[HTML]{EDEDED}{\ul \textbf{0.005}}   & {\color[HTML]{00009B} 11.375} & \cellcolor[HTML]{EDEDED}{\color[HTML]{00009B} 9.506}  & {\color[HTML]{00009B} 10.558} & \cellcolor[HTML]{EDEDED}{\color[HTML]{00009B} 7.544}  & {\color[HTML]{00009B} 5.763}  & \multicolumn{1}{c|}{\cellcolor[HTML]{EDEDED}{\color[HTML]{00009B} {\ul \textbf{1.103}}}} & 7.602     & \multicolumn{1}{c|}{\cellcolor[HTML]{EDEDED}\textbf{2.442}}          & 6.593                         & \cellcolor[HTML]{EDEDED}3.247                                        & 6.545                         & \cellcolor[HTML]{EDEDED}4.442                         \\
\multicolumn{1}{c|}{}                                     & NC $(\times 10^2) \uparrow$               & 98.541          & \cellcolor[HTML]{EDEDED}99.774         & 98.915          & \cellcolor[HTML]{EDEDED}99.796         & 99.045      & \cellcolor[HTML]{EDEDED}\textbf{99.802}     & 99.887   & \cellcolor[HTML]{EDEDED}{\ul \textbf{99.976}}  & {\color[HTML]{00009B} 98.925} & \cellcolor[HTML]{EDEDED}{\color[HTML]{00009B} 99.720} & {\color[HTML]{00009B} 99.236} & \cellcolor[HTML]{EDEDED}{\color[HTML]{00009B} 99.847} & {\color[HTML]{00009B} 99.423} & \multicolumn{1}{c|}{\cellcolor[HTML]{EDEDED}{\color[HTML]{00009B} 99.954}}               & 99.397    & \multicolumn{1}{c|}{\cellcolor[HTML]{EDEDED}\textbf{99.961}}         & 99.382                        & \cellcolor[HTML]{EDEDED}{\ul \textbf{99.961}}                        & 99.397                        & \cellcolor[HTML]{EDEDED}99.951                        \\
\multicolumn{1}{c|}{\multirow{-4}{*}{-0.5\%}}             & Time (s)             & \multicolumn{2}{c}{95.022}                               & \multicolumn{2}{c}{121.724}                              & \multicolumn{2}{c}{172.397}                               & \multicolumn{2}{c||}{360.744}                              & \multicolumn{2}{c}{34.322}                                                            & \multicolumn{2}{c}{45.352}                                                            & \multicolumn{2}{c|}{41.786}                                                                                              & \multicolumn{2}{c|}{69.991}                                                      & \multicolumn{2}{c}{96.635}                                                                           & \multicolumn{2}{c}{55.133}                                                            \\ \hline
\multicolumn{1}{c|}{}                                     & CD $(\times 10^4) \downarrow$             & 0.290           & \cellcolor[HTML]{EDEDED}0.266          & 0.152           & \cellcolor[HTML]{EDEDED}0.137          & 0.081       & \cellcolor[HTML]{EDEDED}\textbf{0.072}      & 0.004    & \cellcolor[HTML]{EDEDED}{\ul \textbf{0.004}}   & {\color[HTML]{00009B} 13.033} & \cellcolor[HTML]{EDEDED}{\color[HTML]{00009B} 15.306} & {\color[HTML]{00009B} 5.957}  & \cellcolor[HTML]{EDEDED}{\color[HTML]{00009B} 6.649}  & {\color[HTML]{00009B} 0.949}  & \multicolumn{1}{c|}{\cellcolor[HTML]{EDEDED}{\color[HTML]{00009B} 1.097}}                & 0.665     & \multicolumn{1}{c|}{\cellcolor[HTML]{EDEDED}\textbf{0.723}}          & 0.587                         & \cellcolor[HTML]{EDEDED}{\ul \textbf{0.635}}                         & 1.048                         & \cellcolor[HTML]{EDEDED}1.085                         \\
\multicolumn{1}{c|}{}                                     & HD $(\times 10^2) \downarrow$              & 0.137           & \cellcolor[HTML]{EDEDED}0.115          & 0.077           & \cellcolor[HTML]{EDEDED}0.048          & 0.048       & \cellcolor[HTML]{EDEDED}\textbf{0.031}      & 0.006    & \cellcolor[HTML]{EDEDED}{\ul \textbf{0.004}}   & {\color[HTML]{00009B} 15.776} & \cellcolor[HTML]{EDEDED}{\color[HTML]{00009B} 16.139} & {\color[HTML]{00009B} 16.357} & \cellcolor[HTML]{EDEDED}{\color[HTML]{00009B} 11.761} & {\color[HTML]{00009B} 4.461}  & \multicolumn{1}{c|}{\cellcolor[HTML]{EDEDED}{\color[HTML]{00009B} 1.243}}                & 7.655     & \multicolumn{1}{c|}{\cellcolor[HTML]{EDEDED}{\ul \textbf{0.411}}}    & 8.837                         & \cellcolor[HTML]{EDEDED}1.097                                        & 7.910                         & \cellcolor[HTML]{EDEDED}\textbf{0.570}                \\
\multicolumn{1}{c|}{}                                     & NC $(\times 10^2) \uparrow$               & 98.593          & \cellcolor[HTML]{EDEDED}99.759         & 98.910          & \cellcolor[HTML]{EDEDED}99.781         & 98.998      & \cellcolor[HTML]{EDEDED}\textbf{99.790}     & 99.864   & \cellcolor[HTML]{EDEDED}{\ul \textbf{99.973}}  & {\color[HTML]{00009B} 97.530} & \cellcolor[HTML]{EDEDED}{\color[HTML]{00009B} 98.899} & {\color[HTML]{00009B} 98.452} & \cellcolor[HTML]{EDEDED}{\color[HTML]{00009B} 99.480} & {\color[HTML]{00009B} 99.497} & \multicolumn{1}{c|}{\cellcolor[HTML]{EDEDED}{\color[HTML]{00009B} 99.927}}               & 99.698    & \multicolumn{1}{c|}{\cellcolor[HTML]{EDEDED}\textbf{99.964}}         & 99.687                        & \cellcolor[HTML]{EDEDED}{\ul \textbf{99.968}}                        & 99.672                        & \cellcolor[HTML]{EDEDED}99.944                        \\
\multicolumn{1}{c|}{\multirow{-4}{*}{0.5\%}}              & Time (s)             & \multicolumn{2}{c}{53.826}                               & \multicolumn{2}{c}{74.450}                               & \multicolumn{2}{c}{126.879}                               & \multicolumn{2}{c||}{312.150}                              & \multicolumn{2}{c}{72.416}                                                            & \multicolumn{2}{c}{86.086}                                                            & \multicolumn{2}{c|}{76.644}                                                                                              & \multicolumn{2}{c|}{151.134}                                                     & \multicolumn{2}{c}{364.322}                                                                          & \multicolumn{2}{c}{106.454}                                                           \\ 
\Xhline{2.3\arrayrulewidth}
\end{tabular}
\end{table*}

\begin{figure*}[h]
    \centering
    \includegraphics[width=0.99\linewidth]{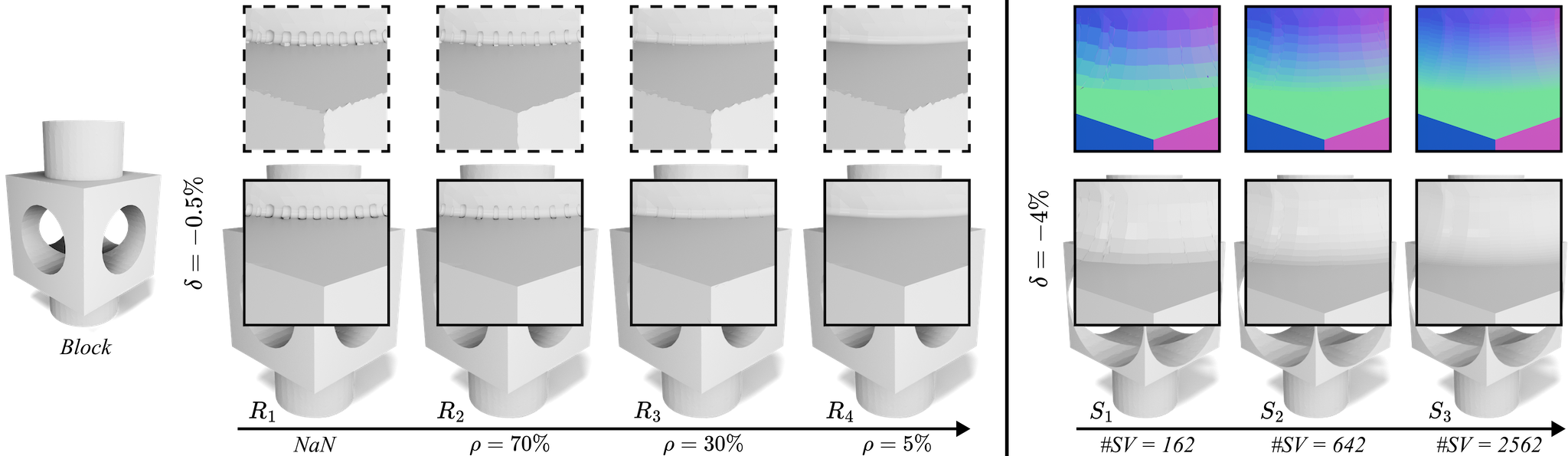}
    \caption{For the ablation study on the block model with sharp features, we use a high-quality mesh with $\#F = 40\text{K}$ and adopt centroids as blue-noise samples. The solid box displays the optimized results, while the dashed box shows the results before optimization. For small offset distances $\delta=-0.5\%$, a smaller $\rho$ can better handle the rounding behavior of sharp regions, reducing the side effects introduced by excessive samples from the discrete sphere \textit{(left)}. For large offset distances $\delta=-4\%$, a more finely discretized sphere with more vertices can provide \ZH{more natural and accurate} results \textit{(right)}.}
    \label{fig:ablation_block}
\end{figure*}



In our method, sampling plays a crucial role in determining the quality of the results. Specifically, our approach relies on four key components: blue noise sampling, clearance parameter \(\rho\), the level of 1vN sampling, and the dihedral angle threshold for detecting sharp feature lines. 


\textbf{Number of Blue-Noise Sampling.}
Blue-noise sampling uniformly captures the overall shape of the input model. To evaluate its effectiveness, we tested it on a smooth and detailed shape, \textit{Kitten} (Figure~\ref{fig:ablation_kitten_blue}). In general, denser sample points yield better results, regardless of whether the offset distance is small or large (Table~\ref{tb:ablation_study}, left). However, this improvement comes at the cost of increased computation time.

\begin{figure}[!t]
    \centering
    \includegraphics[width=0.99\linewidth]{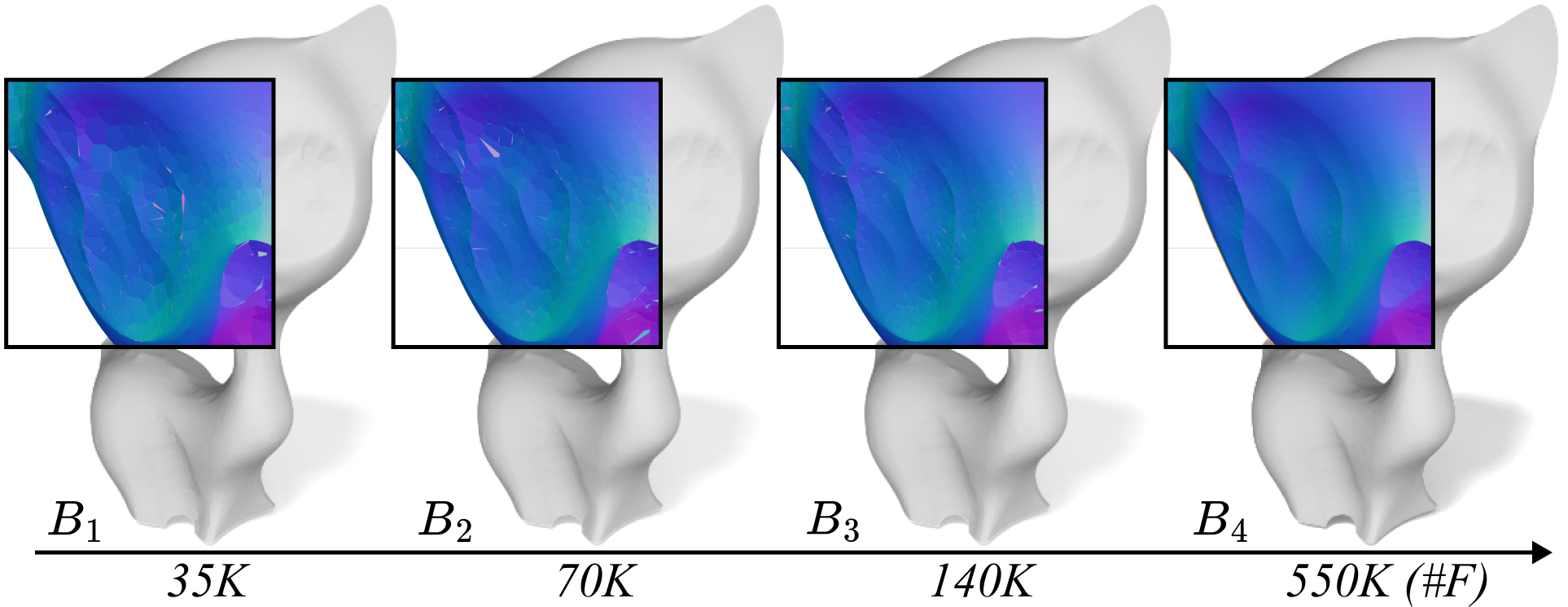}
    \caption{Different numbers of blue noise samples for the kitten model at a distance of $\delta = -4\%$. With more samples, the results are \ZH{more natural and accurate}.}
    \label{fig:ablation_kitten_blue}
\end{figure}

\textbf{Clearance Parameter \(\rho\).}
As shown in Table~\ref{tb:ablation_study}, this approach is particularly effective for small offset distances, with accuracy improving significantly as \(\rho\) decreases. Figure~\ref{fig:ablation_block} (left) illustrates the results for \(\delta = -0.5\%\). 

Without setting safe bands (dark-green region in the inset figures in Section~4.1), excessive sampling from the discrete sphere leads to noticeable protrusions. Even with misalignment elimination, these artifacts cannot be fully corrected. However, reducing \(\rho\) results in \ZH{natural} regions, even without the elimination step. In addition, concave regions can retain sharpness through optimization, and an appropriate choice of \(\rho\) can help preserve desirable features even before optimization is applied.
We empirically found that smaller offset distances benefit from a smaller \(\rho\), which better preserves sharp geometric features, while larger offsets can tolerate a relatively larger \(\rho\) without introducing noticeable artifacts.

\textbf{Level of 1vN Sampling.} 
\ZH{
For constant-radius offsets, we adopt a hybrid sampling strategy that combines spherical sampling at vertices with Slerp-based 1vN interpolation along edges. As described in Section~4.1, we place a discretized Gaussian spherical surface at each vertex to generate 1vN samples, while Slerp interpolation is employed for edge sampling. This design enables the generation of rounded offset geometry around both edges and vertices.
}

To evaluate the effect of the discretization level, we associate each Slerp interpolation with the corresponding resolution of the sphere. As shown in Table~\ref{tb:ablation_study}, increasing the discretization level consistently improves accuracy across all offset distances, with more pronounced effects observed at larger distances. Figure~\ref{fig:ablation_block} (right) illustrates the case for \(\delta = -4\%\). As more vertices are incorporated into the spherical surface, the rounded regions become \ZH{more natural}.

\begin{figure*}[!t]
    \centering
    \includegraphics[width=0.99\linewidth]{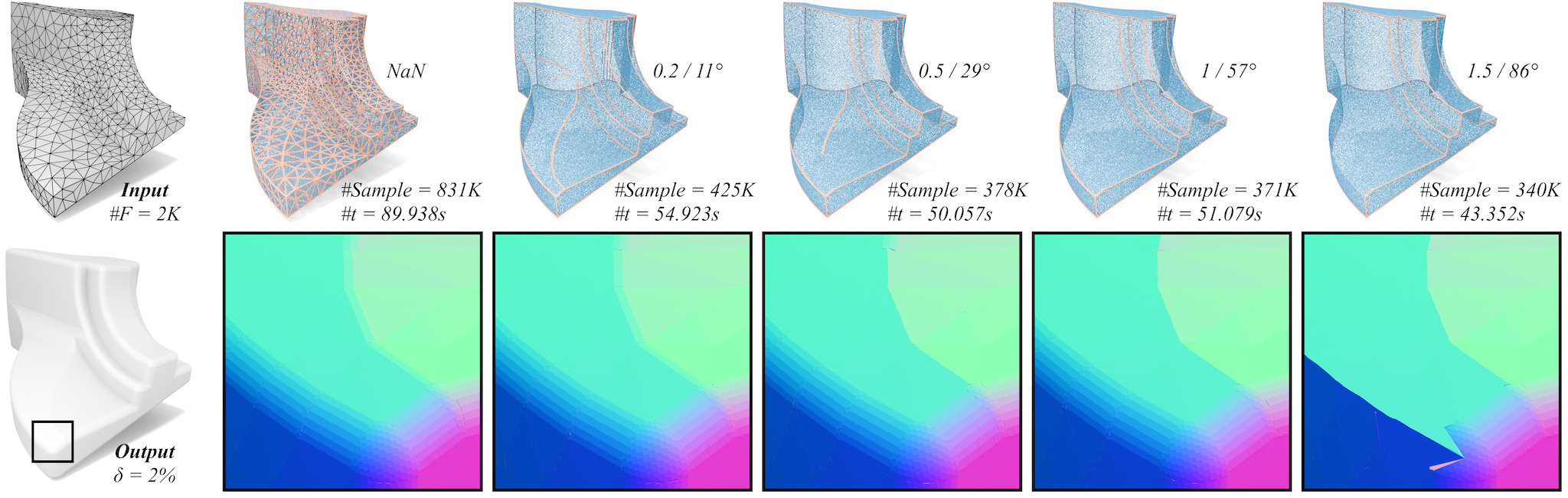}
    \caption{Ablation study on the Fandisk model with dihedral angle threshold for filtering. }
    \label{fig:ablation_feature}
    \vspace{-2mm}
\end{figure*}

\begin{table}[!t] \scriptsize
\centering
\caption{Ablation study on the Fandisk model evaluating the effects of our dihedral angle-based Convexity-Based Filtering.
\label{tb:ablation_convex}}
\vspace{0   mm}
\renewcommand{\arraystretch}{1.2}
\begin{tabular}{c|ccccc}
\Xhline{2.3\arrayrulewidth}
\textbf{Dihedral} $(rad/^{\circ})$ & \textbf{NaN}    & \textbf{0.2 / 11°} & \textbf{0.5 / 29°} & \textbf{1 / 57°} & \textbf{1.5 / 86°} \\ 
\hline
All Samples  & 831K   & 425K      & 378K      & 371K    & 340K      \\
CD $(\times 10^4) \downarrow$        & 0.099  & 0.105     & 0.134     & 0.175   & 2.213     \\
HD $(\times 10^2) \downarrow$        & 0.040  & 0.042     & 0.106     & 0.491   & 4.446     \\
NC $(\times 10^2) \uparrow$        & 99.959 & 99.944    & 99.936    & 99.924  & 99.341    \\
Time (s) & 89.938 & 54.923    & 50.057    & 51.079  & 43.352    \\ 
\Xhline{2.3\arrayrulewidth}
\end{tabular}
\end{table}

\textbf{Dihedral Angle Threshold for Identifying Feature Lines.}
To reduce computational overhead, we optionally pre-detect sharp features using a dihedral angle threshold, and apply the 1vN strategy only to those regions.

To evaluate its effectiveness, we conduct an ablation study using the \emph{Fandisk} model, which features intricate sharp structures, as shown in Figure~\ref{fig:ablation_feature}. All evaluation metrics and experimental settings are summarized in Table~\ref{tb:ablation_convex}.

Results indicate that disabling the filtering improves output quality, but at the cost of significantly increased computation time. As the dihedral angle threshold increases, fewer samples are retained around edges and vertices. This reduction lowers computational cost but leads to a decline in quality, particularly as the threshold approaches a right angle. To balance quality and efficiency, we adopt a dihedral angle threshold of 0.2 radians in all our experiments.

\section{Analysis of Clearance Around Mesh Edges}
\label{sec:appendix save_band_size}

\begin{figure}[!t]
    \centering
    \includegraphics[width=0.96\linewidth]{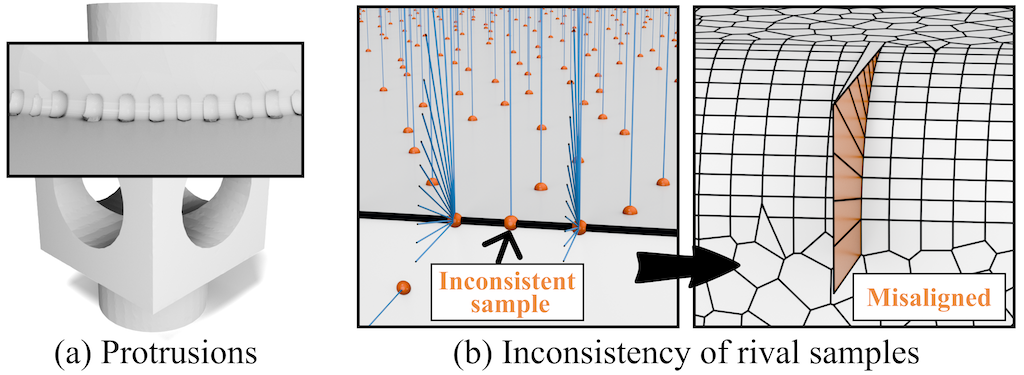}
    \caption{
    Two types of undesirable behaviors affect the quality of the offset surface: (a) Protrusions caused by redundant 1vN sampling. (b) Structural misalignment caused by local inconsistency in the behavior of a 1v1 sample near the 1vN region.
}
    \label{fig:bad_situations}
    \vspace{-5mm}
\end{figure}

\begin{figure}[!t]
    \centering
    \includegraphics[width=0.96\linewidth]{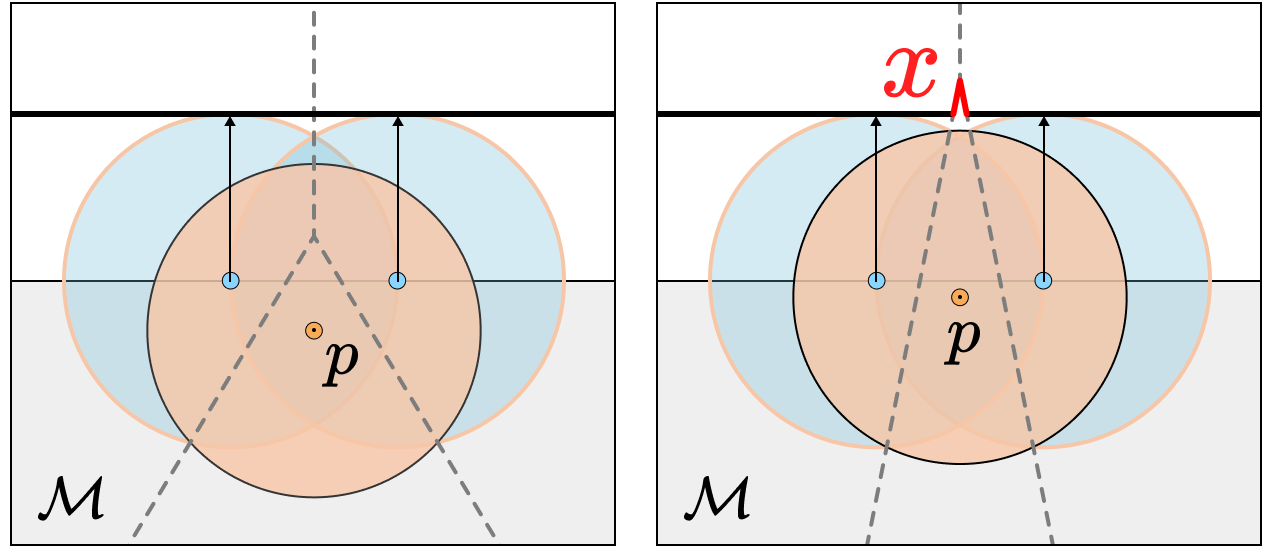}
    \caption{
    2D illustration of Figure~\ref{fig:bad_situations}. When the point $p$ is close to the edge and has significantly different directions, misalignments arise between $p$ and the on-edge Slerp $\epsilon$-samples.
}
\vspace{-5mm}
    \label{fig:safe_band_mis}
\end{figure}

As demonstrated in Figure~\ref{fig:safe_band_mis}, when the point $p$ is close to an edge and exhibits significantly different directions, misalignments can occur between $p$ and the on-edge Slerp $\epsilon$-samples.
These misalignments may result in incorrect adjacency relationships, leading to uncontrolled artifacts and negatively impacting the normal-based optimization.

\begin{wrapfigure}{r}{0.35\linewidth}
\vspace{2mm}
\includegraphics[width=.99\linewidth]{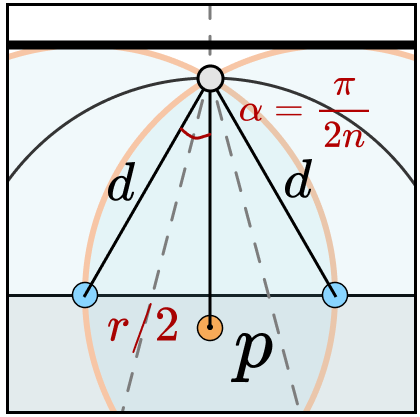}
\label{fig:band_size}
\vspace{0mm}
\end{wrapfigure} 
To prevent this issue, as illustrated in the inset figure, we consider the critical configuration in which the offset ball centered at point \(p\) just touches the intersection point of two on-edge samples, with the inter-sample distance denoted as \(r\). In the constant-radius case, \ZH{the size of safe band (dark-green region in the inset figures in Section 4.1), which is controlled by clearance parameter $\rho$}, is determined by the relation \(d \cdot (1 - \cos\alpha) \in (0, \frac{r}{2})\), making \(\frac{r}{2}\) a conservative and reliable choice. For the variable-radius case, as the samples are densely distributed and the offset distances vary smoothly, \(\frac{r}{2}\) is also generally a safe and effective choice.

\section{Misalignment Elimination Optimization}
\label{sec:optimization_terms}
The regularized optimization terms, using the homogeneous coordinate system, can be expressed as follows:
\begingroup 
\scriptsize
\begin{align}
E(v) =& \sum_{i=1}^k \left((v - p_i)\cdot \boldsymbol{n}_i - d\right)^2 + \lambda \|v - v_0\|^2 \\
= & \sum_{i=1}^k \left((v - p_i)\cdot \boldsymbol{n}_i - d\right)^2 
+ \lambda \left( v^T v - 2 v^T v_0 + v_0^T v_0 \right) \\
= & \sum_{i=1}^k 
    \left( \begin{bmatrix}
        v \\ 1
    \end{bmatrix}^T \begin{bmatrix}
        \boldsymbol{n}_i \\ -p_i \cdot \boldsymbol{n}_i - d_i
    \end{bmatrix} \right)^2 \nonumber \\
& \quad + 
    \lambda \left( 
    \begin{bmatrix}
        v \\ 1
    \end{bmatrix}^T \begin{bmatrix}
        I & 0 \\ 0 & 0
    \end{bmatrix} \begin{bmatrix}
        v \\ 1
    \end{bmatrix} 
    + \begin{bmatrix}
        v \\ 1
    \end{bmatrix}^T 
    \begin{bmatrix}
        2v_0 \\ v_0^T v_0
    \end{bmatrix} \right) \\
= & \begin{bmatrix}
        v \\ 1
    \end{bmatrix}^T \left(
    \sum_{i=1}^k 
    \begin{bmatrix}
        \boldsymbol{n}_i \\ -p_i \cdot \boldsymbol{n}_i - d_i
    \end{bmatrix}
    \begin{bmatrix}
        \boldsymbol{n}_i \\ -p_i \cdot \boldsymbol{n}_i - d_i
    \end{bmatrix}^T
     + \begin{bmatrix}
        \lambda I & 0 \\ 0 & 0
    \end{bmatrix} \right)
    \begin{bmatrix}
        v \\ 1
    \end{bmatrix} \nonumber \\
& \quad + 
    \begin{bmatrix}
        v \\ 1
    \end{bmatrix}^T
    \begin{bmatrix}
        -2\lambda v_0 \\ \lambda v_0^T v_0
    \end{bmatrix} \\
\Rightarrow E(\tilde{v})&= \tilde{v}^T \boldsymbol{\tilde{H}} \tilde{v} + \boldsymbol{\tilde{b}}^T \tilde{v}
\end{align}
\endgroup
This optimization problem can be solved by converting it into the equation $E'(\tilde{v}) = 0$, the solution is derived as follows:
\begingroup
\scriptsize
\begin{align}
    & E'(\tilde{v}) = 0 \\
    & \Rightarrow 2\boldsymbol{\tilde{H}} \tilde{v}  = -\boldsymbol{\tilde{b}} \\ 
    & \Rightarrow 2 \Bigg(
    \sum_{i=1}^k 
    \begin{bmatrix}
        \boldsymbol{n}_i \boldsymbol{n}_i^T & ( -p_i \cdot \boldsymbol{n}_i - d_i) \boldsymbol{n}_i \\[6pt] 
        ( -p_i \cdot \boldsymbol{n}_i - d_i) \boldsymbol{n}_i^T & ( -p_i \cdot \boldsymbol{n}_i - d_i)^2
    \end{bmatrix} 
    + 
    \begin{bmatrix}
        \lambda I & 0 \\[6pt]
        0 & 0
    \end{bmatrix} \Bigg)
    \begin{bmatrix}
        v \\ 1
    \end{bmatrix} \nonumber \\
    & \quad = 
    \begin{bmatrix}
        2\lambda v_0 \\ -\lambda v_0^T v_0
    \end{bmatrix} \\
    &\Rightarrow 
    \begin{bmatrix}
        \sum_{i=1}^k \boldsymbol{n}_i\boldsymbol{n}_i^T + \lambda I & \sum_{i=1}^k(-p_i \cdot \boldsymbol{n}_i - d_i)\boldsymbol{n}_i
    \end{bmatrix}\begin{bmatrix}
        v \\ 1
    \end{bmatrix} = \lambda v_0 \\
    &\Rightarrow 
    \left( \sum_{i=1}^k \boldsymbol{n}_i\boldsymbol{n}_i^T + \lambda I\right)v = \lambda v_0 + \sum_{i=1}^k(p_i \cdot \boldsymbol{n}_i + d_i)\boldsymbol{n}_i
\end{align}
\endgroup
Denote $\boldsymbol{H} = \sum_{i=1}^k \boldsymbol{n}_i\boldsymbol{n}_i^T + \lambda I$, $\boldsymbol{b}=\sum_{i=1}^k(p_i \cdot \boldsymbol{n}_i + d_i)\boldsymbol{n}_i$, the solution is $v = \boldsymbol{H}^{-1}\left(\lambda v_0 + \boldsymbol{b}\right)$. Since $\boldsymbol{n}_i\boldsymbol{n}_i^T$ is positive semi-definite and $I$ is positive definite, $H$ remains positive definite (and thus invertible) for any $\lambda > 0$. The $\lambda$ term is introduced to address the non-invertable issue by incorporating the initial vertex positions. As long as a small $\lambda > 0$ is chosen, the results remain stable, as its primary role is to ensure matrix invertibility.

\begin{figure*}[!htbp]
    \centering
    \includegraphics[width=.99\linewidth]{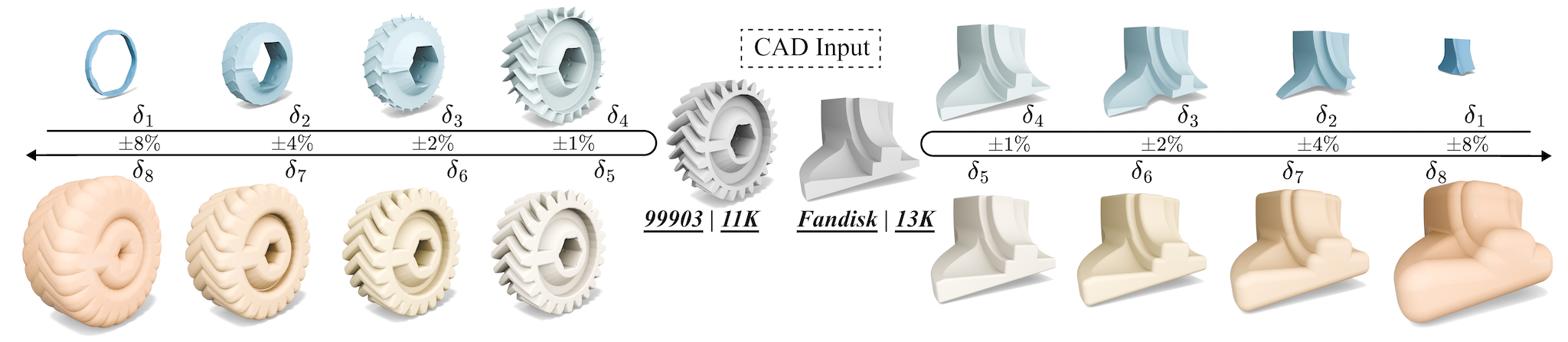}
    \vspace{-2mm}
    \caption{Qualitative results of the constant-radius offset of CAD models at various offset distances. The behaviors of sharp features in the offset are handled correctly, preserving sharpness in the inward offsets and roundness in the outward offsets. }
    \label{fig:robust_cad}
\end{figure*}

\begin{figure*}[!htbp]
    \centering
    \includegraphics[width=.99\linewidth]{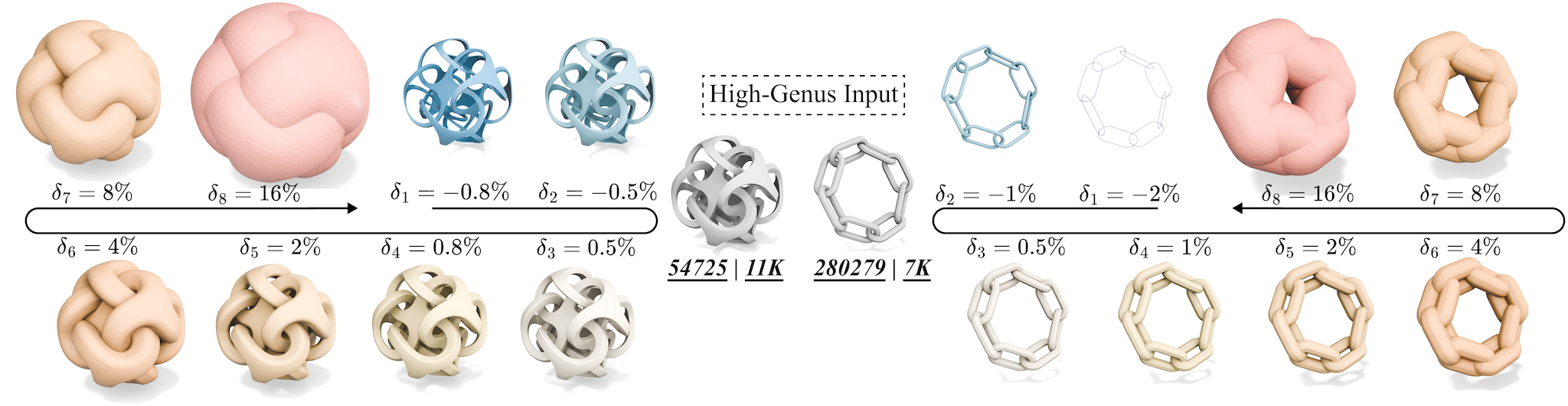}
    \vspace{-2mm}
    \caption{Qualitative results of the constant-radius offset of high-genus models at various offset distances. Topology changes are preserved correctly, increasing or maintaining the genus in inward offsets and reducing or maintaining it in outward offsets.}
    \label{fig:robust_genus}
\end{figure*}

\begin{figure*}[!htbp]
    \centering
    \includegraphics[width=.99\linewidth]{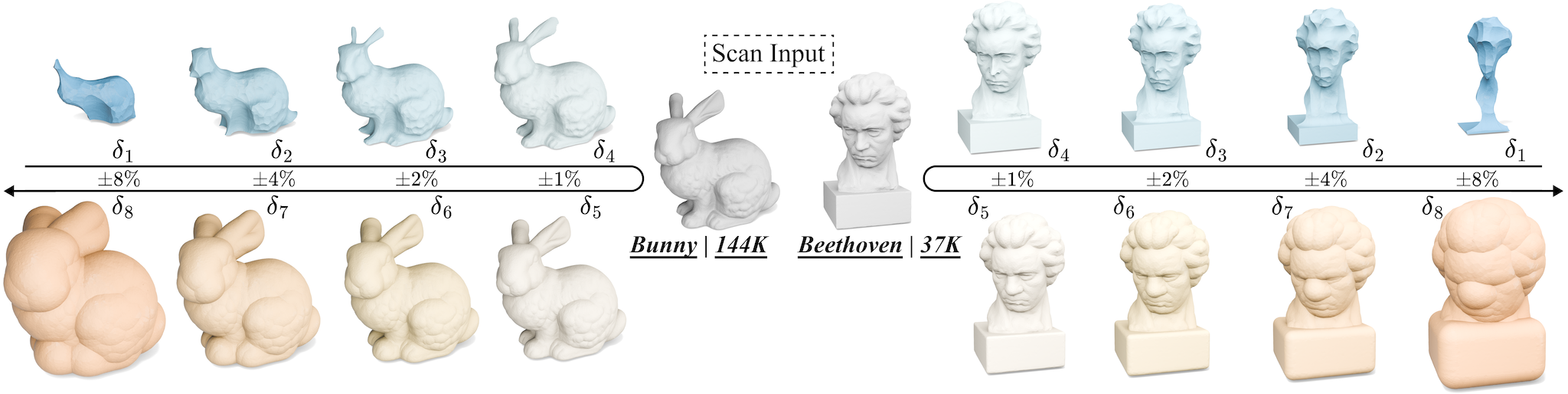}
    \vspace{-2mm}
    \caption{Qualitative results of the constant-radius offset of real scan models at various offset distances. The details are captured and processed correctly, and the features appearing in the offset surfaces are well preserved.}
    \label{fig:robust_scan}
\end{figure*}

\section{Statistics of Output Mesh Quality}
\label{sec:Statistics of Output Mesh Quality}

We assessed the quality of 7,609 output meshes in terms of manifoldness, self-intersections, and boundaries using evaluation functions from libigl. To utilize these existing functions, we triangulated the output polygonal facets. For each output, we computed the \textbf{proportions} of non-manifold edges (EN) and boundary edges (EB), non-manifold vertices (VN), and non-coplanar intersecting faces (FI) relative to the total number of edges, vertices, and faces. Table~\ref{tb:output quality} shows the statistics. 

\ZH{
As described in Section~7, the triangular faces are generated using the ear clipping algorithm, followed by a post-processing step that removes duplicate vertices and degenerate faces. 
Due to this polygon-based triangulation pipeline, a small number of non-manifold edges or vertices, boundary edges, and occasional self-intersections may arise. 
However, as demonstrated in Section~5.2, these artifacts have negligible impact on the geometric accuracy of the resulting offset surfaces.
}

\begin{table}[htbp]
\vspace{5mm}
\centering
\caption{Statistics for $\delta=2\%$ and $\delta=-2\%$ with raw and optimized results (ours/ours+).\label{tb:output quality}}
\scriptsize
\setlength{\tabcolsep}{1.1pt}
\renewcommand{\arraystretch}{1.2}
\begin{tabular}{c|c|cccc}
\Xhline{2.3\arrayrulewidth}
$\delta$ & Metric & EN & EB & VN & FI \\
\hline
\multirow{3}{*}{$2\%$} 
& min  & 0\%/0\%         & 0\%/0\%         & 0\%/0\%          & 0.001\%/2.128\% \\
& mean & 0\%/0.004\%     & 0\%/0\%         & 0.023\%/0.028\%  & 0.056\%/9.180\% \\
& max  & 0.001\%/0.528\% & 0\%/0\%         & 0.302\%/0.540\%  & 6.062\%/23.138\%\\
\hline
\multirow{3}{*}{$-2\%$}
& min  & 0\%/0\%         & 0\%/0\%         & 0\%/0\%          & 0\%/0.249\% \\
& mean & 0\%/0.002\%     & 0\%/0\%         & 0.022\%/0.023\%  & 0.076\%/11.759\% \\
& max  & 0.002\%/0.376\% & 0.0003\%/0.0003\% & 0.367\%/0.398\% & 5.128\%/66.176\% \\
\Xhline{2.3\arrayrulewidth}
\end{tabular}
\end{table}

\ZH{
\section{Results of Constant-Radius Offsets at Different Distances}
\label{sec:robust_distances}

We selected three types of models to compute offsets at various distances: CAD models (Figure~\ref{fig:robust_cad}), high-genus models (Figure~\ref{fig:robust_genus}) and real-scan models (Figure~\ref{fig:robust_scan}).  These results clearly demonstrate that our method robustly computes offsets for different model types across a wide range of distances, including both small and large. 

}

\section{Feature-Preserved Offset of Triangle Meshes}
\begin{figure}[!t]
    \centering
    \includegraphics[width=0.98\linewidth]{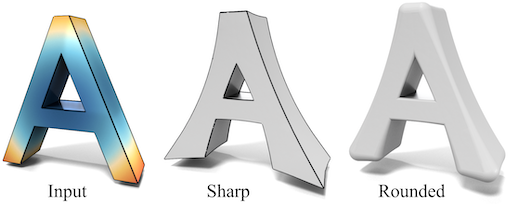}
\caption{
Given the letter ``A'' \textit{(left)}, OffsetCrust can produce either a feature-preserved offset \textit{(middle)} or a rounded-feature offset \textit{(right)} by controlling the displacement directions when generating edge-around and vertex-around displaced points.
}
    \label{fig:A_offset}
\end{figure}

\begin{figure}[!t]
    \centering
    \includegraphics[width=0.98\linewidth]{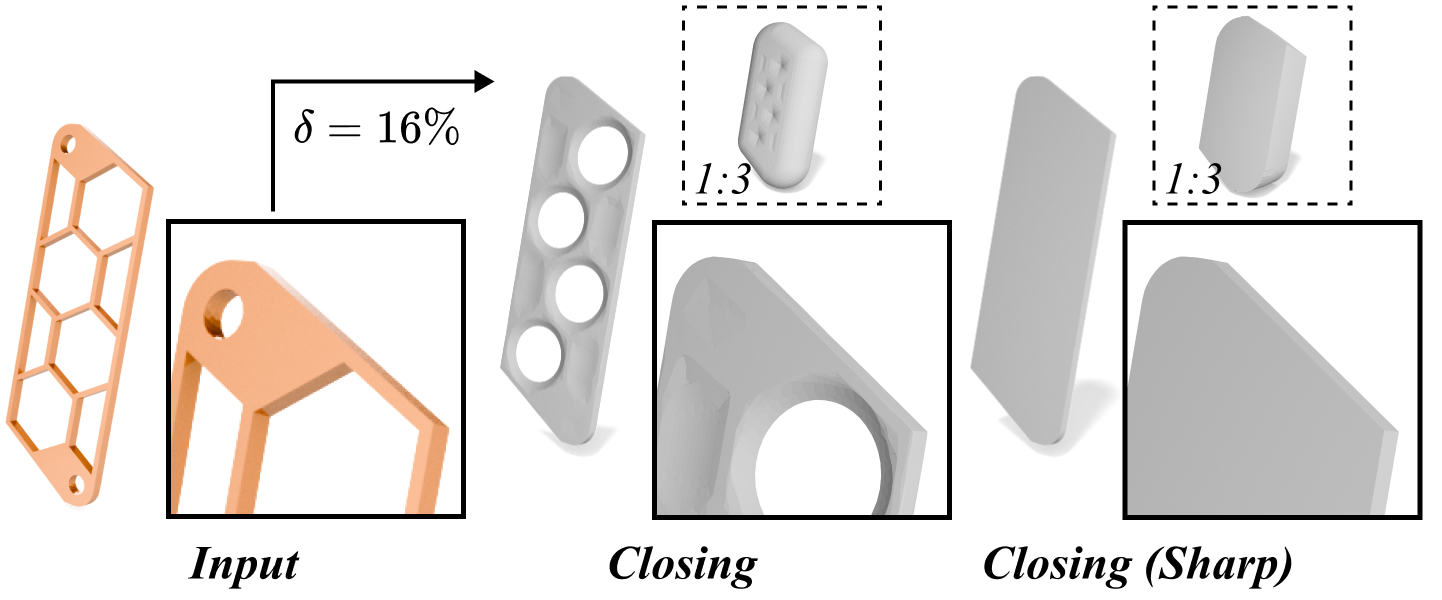}
    \caption{
\ZH{
A closing operation consists of dilation followed by erosion. By controlling the displacement directions, our \textit{OffsetCrust} can produce both standard smooth outcomes and sharp-featured variants.
Notably, this feature-preserved offset can modify the topology of the shape resulting in a genus-zero geometry.
}
\vspace{-2mm}
\label{fig:opening and closing}}
\end{figure}

Recall that when generating edge-around and vertex-around base samples, we employ a Slerp-like technique to produce displaced points. In fact, we can further restrict the set of allowable displacement directions. In extreme cases, intermediate directions can be skipped entirely, resulting in a \emph{feature-preserved} offset shape. For example, in Figure~\ref{fig:A_offset}, we use the letter ``A'' as input and generate both a feature-preserved offset and a rounded-feature offset.

We extend this behavior to morphological operations such as closing, which consists of dilation followed by erosion. By similarly controlling the displacement directions, our \textit{OffsetCrust} can produce either standard smooth outcomes or sharp-featured variants. Notably, this feature-preserved offset can change the genus of the shape (see Figure~\ref{fig:opening and closing}). \ZH{The intermediate results are remeshed using TetWild to serve as inputs for the subsequent stage.}

\end{document}